%% file: template.tex
\documentclass{article}
\input{preamble}

\title{The profit--bias identity in sports betting: bookmaker profit as the public's prediction error}

\author{
 Jacek P. Dmochowski \\
  Department of Biomedical Engineering\\
  City College of New York\\
  New York, NY 10031 \\
  \texttt{jdmochowski@ccny.cuny.edu} 
}

\begin{document}
\maketitle

\input{sections/abstract}
\input{sections/introduction}
\section*{Results}
\input{sections/terminology}   
\input{sections/identity}      
\input{sections/belief.tex}    
\input{sections/goldilocks}    
\input{sections/schematic}     
\input{sections/empirical}     

\input{sections/discussion}

\input{sections/methods}       

\bibliographystyle{plainnat}
\bibliography{references}

\input{sections/supplement}

\end{document}

%% file: preamble.tex
\usepackage{arxiv}
\usepackage{natbib}          
\usepackage[utf8]{inputenc} 
\usepackage[T1]{fontenc}    
\usepackage{hyperref}       
\usepackage{url}            
\usepackage{booktabs}       
\usepackage{amsfonts}       
\usepackage{nicefrac}       
\usepackage{microtype}      
\usepackage{lipsum}
\usepackage{graphicx}
\graphicspath{ {./images/} }
\usepackage{xcolor}

\usepackage{amsmath,amssymb}
\usepackage{amsthm}
\usepackage{braket}

\theoremstyle{plain}
\newtheorem{theorem}{Theorem}

\newtheorem{corollary}[theorem]{Corollary}
\newtheorem{proposition}[theorem]{Proposition}

\theoremstyle{definition}

\theoremstyle{remark}

%% file: sections/abstract.tex
\begin{abstract}
Sports betting moves money continuously from a large public to a small
number of firms. The most influential account of that flow, due to Levitt (2004),
holds that books price away from the market-clearing point to exploit predictable
public biases. It computes profit from two numbers (the probability that a side wins the
proposition and the fraction of handle it attracts), treating the share on a side as
independent of the outcome. Here we relax that assumption and derive a
\emph{profit--bias identity}: profit is affine and increasing in the expected share of
handle on the losing side, with Levitt's expression as the special case of
independence. The identity resolves the book's margin into exactly three channels: its hold,
the product of its price shading and the public's lean, and the covariance between bet share
and outcome. A lean is thus worthless without shading, and shading is worthless
without a lean. Under a public-belief model, the profit driver is the public's Bayes error: the probability of a representative bettor selecting the
losing side. We provide necessary and sufficient conditions for a ``Goldilocks
Zone'': prices at which book \emph{and} bettor both profit. Testing these predictions
on $1{,}139$ Major League Baseball games, we find that the apparent dependence between bet
share and outcome is a Simpson's paradox: present when games are pooled, but absent once
they are separated by which side the book favored. The public leans heavily toward favorites, but we
detect no matching shading, and the realized margin is indistinguishable from the hold.
\end{abstract}

%% file: sections/introduction.tex
\section*{Introduction}

Every account of how a sportsbook makes money rests on a claim about where the betting public's money goes relative to the price at which it was offered. The most influential, due to \citet{levitt2004gambling}, posits that books deliberately price away from the market-clearing value to exploit predictable biases in public betting as opposed to balancing their exposure. That model, and the substantial literature testing it \citep{paul2007sportsbook, paul2008nba, paul2009ncaa, paul2011nfl, paul2012nhl, humphreys2010pointspread, humphreys2011financial}, computes the book's expected profit from two numbers: the probability that the favorite covers, and the fraction of the handle it attracts. Two numbers suffice, as we show, only if the money and the outcome are statistically independent---only if the side drawing the handle is neither more nor less likely to cover than its posted price implies. To the best of our knowledge, this assumption has gone unexamined.

Treating the public's money as unconnected to the outcome discards precisely the information that separates an informed betting public from an uninformed one, namely whether the side that attracts the money is more or less likely to win. This distinction is not incidental: a competing account attributes the bookmaker's margin not to public bias but to adverse selection against privately informed bettors~\citep{shin1991optimal,shin1992prices,shin1993measuring}, a mechanism recently re-examined and weighed against ordinary bettor disagreement~\citep{whelan2024risk,whelan2025estimates}. By assuming independence between bet share and outcome, Levitt's model cannot distinguish a \emph{sharp} public, whose money follows the eventual winner (as the insider account would predict), from a \emph{square} one, whose money follows the eventual loser.


Here we revise the Levitt model to allow a statistical dependence between betting behavior and match outcomes. Relaxing the independence restriction yields a \emph{profit--bias identity}: the sportsbook's expected profit is an affine function of the expected share of handle on the \emph{losing} side, and it reduces to Levitt's expression when bet shares are independent of the outcome. Because the identity presumes neither account, the sharp and square cases can be distinguished empirically rather than assumed away~\citep{woodland1994baseball,gandar2002reexamining,ottaviani2008favorite}. The identity further separates the book's expected margin into three parts: the hold that a balanced book collects, a term in which the book's price shading multiplies the public's lean, and the covariance between bet share and outcome. Because the second term is the product of shading and lean, neither factor generates profit on its own: a lean earns nothing for a book that does not shade, and shading earns nothing when the public does not lean. Under a public-belief model in which bettors act on a noisy, biased signal of the game~\citep{snowberg2010explaining,jullien2000estimating,avery1999sentiment}, the profit-driving quantity is the probability that the public selects the losing side, a misclassification (Bayes-error) probability. The same formulation characterizes a ``Goldilocks Zone'' of prices at which the sportsbook and an infinitesimal bettor can simultaneously profit~\citep{shleifer1997limits,hubacek2019exploiting}.


We tested these predictions with data from $1{,}139$ Major League Baseball (MLB) games, finding that the apparent bias of public money onto the eventual losing side is a Simpson's paradox: present when games are pooled, and absent once they are separated by which side the book has designated as the favorite. The public is square in the aggregate but not in either of its halves. The data do exhibit a strong favorite bias in public betting, but \emph{without} an accompanying increase in book profit over its hold, which we connect to the absence of detectable price shading. 

%% file: sections/terminology.tex
\subsection*{Terminology}
The forthcoming treatment employs the terminology of sports betting, which is reviewed in Table~\ref{tab:terminology}.
Of note, the terms \emph{lean} and \emph{shading} refer to skews in public betting and sportsbook pricing, respectively, and are often conflated.

\begin{table}[t]
\centering
\caption{\textbf{Betting terminology.} The terms are defined in the point-spread setting, where $M$ is the realized outcome, $b_h$ the public's home share of handle, $p^\ast$ the price-implied probability that the home side wins the bet (with the book's commission removed), and $\bar F_m=\Pr(M>s)$ is the corresponding true
probability.}
\label{tab:terminology}
\small
\begin{tabular}{@{}l l p{0.66\textwidth}@{}}
\toprule
\textbf{Term} & \textbf{Symbol} & \textbf{Meaning} \\
\midrule
Spread   & $s$      &  A number posted by the book. The spread is a handicap applied to one side's score so that a wager on either is closer to a 50/50 proposition. The
                       bettor's task is to predict whether the outcome will exceed the
                       spread: $M>s$. \\[2pt]
Price    & $\phi_h,\phi_v$ & The profit on a winning unit wager on the home and away sides,
                       respectively. At the standard $-110$ price $\phi_h=\phi_v=100/110=0.909$:
                       the bettor risks $1.10$ to win $1$. The special case $\phi_h=\phi_v\equiv\phi$ is referred to herein as
                       \emph{symmetric vig}. \\[2pt]
Cover    & $M>s$    & A side ``covers'' when it wins the match taking into account the spread -- also called ``beating the spread''. We take the home side to
                       cover when $M>s$ and to lose the bet when $M<s$. \\[2pt]
Push     & $M=s$    & A ``push'' occurs when the realized outcome is equal to the spread. In this event, all wagers are returned.  \\[2pt]
Handle   & ---      & The total amount staked. Book profit is reported per unit of
                       handle, such that a margin of $4\%$ is four cents kept per dollar
                       wagered. \\[2pt]
\addlinespace
Vig      & ---      & The book's commission, embedded in the prices rather than charged
                       separately. The vig is the amount by which the sum of the two
                       price-implied probabilities exceeds one: $\left( 1 + \phi_h \right)^{-1} + \left( 1 + \phi_v \right)^{-1} - 1$. For $\phi_h = \phi_v = \phi$, the vig is
                       $(1-\phi)/(1+\phi)$, or $4.8\%$ at $-110$. Note that ``the vig'' is often used
                       inconsistently in practice: sometimes as this excess, sometimes as the $10\%$ surcharge on the
                       stake, or sometimes as the hold below. \\[2pt]
Hold     & $h$      & The share of handle that the book keeps when its liabilities are
                       balanced: its return from the vig alone,
                       $h=(1-\phi_h\phi_v)/(2+\phi_h+\phi_v)$, or $(1-\phi)/2$ under a symmetric
                       vig, giving $4.5\%$ at $-110$. \\[2pt]
Shading  & $\theta$ & The deviation of the book's posted price from the truth,
                       $\theta=p^\ast-\bar F_m$. Positive when the home side is priced as more
                       likely to win than it truly is. \\[2pt]
\addlinespace
Lean     & $\lambda$& The public's displacement from balanced betting:
                       $\lambda=b_h-p^\ast$. Positive when public betting favors the home
                       side beyond what is implied by the book price. \\[2pt]
Square   & $\delta>0$ & A public whose money tends to be wagered on the eventual \emph{loser},
                       so that the book profits beyond its hold. \\[2pt]
Sharp    & $\delta<0$ & A public whose money tends to be wagered on the eventual \emph{winner} (book profits less than the hold, or takes a loss). 
                       \\
\bottomrule
\end{tabular}

\vspace{4pt}
\begin{minipage}{0.92\textwidth}
\footnotesize We note also that the term ``line'' is the generalization of the spread $s$ to
markets such as the point total or player props. In all cases it is a number set by the book
that dictates the threshold determining which side wins the wager.
\end{minipage}
\end{table}

%% file: sections/identity.tex
\subsection*{Setup and notation}
We consider a generic two-sided sportsbook proposition, encompassing point spreads, totals, and player props.
Let the realized outcome be a real-valued random variable $M\in\mathbb{R}$. For spreads, $M$ can be the realized
margin of victory (home score minus visitor score); for totals, $M$ can be total points; for props, $M$ can be a
player performance statistic.

The sportsbook posts a line $s\in\mathbb{R}$. Relative to $s$, we write the possible outcomes as:
\[
\text{home loses} \iff M<s,\qquad
\text{visitor loses} \iff M>s,\qquad
\text{push} \iff M=s.
\]
The outcome CDF at $s$ and push probability are given by:
\begin{equation}
\label{eq:Fm_methods}
F_m(s)\equiv\Pr(M<s),
\qquad
\pi_{\mathrm{push}}(s)\equiv\Pr(M=s),
\end{equation}
such that the complementary CDF is $\bar F_m(s)\equiv\Pr(M>s)=1-F_m(s)-\pi_{\mathrm{push}}(s)$.

Let the (random) fraction of total stake placed on the home side at price $s$ be $B_h(s)\in[0,1]$, and let the
fraction placed on the visitor side be $B_v(s)\equiv 1-B_h(s)$. We explicitly allow statistical dependence between
betting behavior and the realized outcome (i.e., $B_h(s)$ may be informative about $M$). We define the conditional
expected bet shares on the \emph{eventual losing side} by:
\begin{equation}
\label{eq:b_def_general}
b_{h,L}(s)\equiv \mathbb{E}[\,B_h(s)\mid M<s\,],
\qquad
b_{v,L}(s)\equiv \mathbb{E}[\,B_v(s)\mid M>s\,].
\end{equation}
These quantities capture the fraction of the handle that (in expectation) ends up on the losing side, conditional on
which side actually loses.

\subsection*{The profit--bias identity}
We derive the per-game expected profit in the symmetric-vig ($\phi_h=\phi_v=\phi$), no-push case ($\pi_{\mathrm{push}}(s)=0$); the general treatment allowing asymmetric vigs
$\phi_h,\phi_v\in[0,1]$ and pushes is deferred
to the \hyperref[sec:si]{SI} (Sec.~S1). Throughout, total stake is normalized to $B_h(s)+B_v(s)=1$. The book's single-game profit is
$\pi(s)=B_h(s)-\phi\,B_v(s)$ when the home side loses, $\pi(s)=B_v(s)-\phi\,B_h(s)$ when the visitor loses, and $\pi(s)=0$ in the event of a push. Conditioning on the realized outcome and taking total expectation gives:
\begin{equation}
\label{eq:profit_total_expectation}
\mathbb{E}[\pi(s)]
=
F_m(s)\,\mathbb{E}\!\left[B_h(s)-\phi\,B_v(s)\mid M<s\right]
+
\bar F_m(s)\,\mathbb{E}\!\left[B_v(s)-\phi\,B_h(s)\mid M>s\right].
\end{equation}
Using $B_v=1-B_h$ and the conditional losing-side shares \eqref{eq:b_def_general}, the two conditional
expectations equal $(1+\phi)\,b_{h,L}(s)-\phi$ and $(1+\phi)\,b_{v,L}(s)-\phi$, respectively, so that:
\begin{equation}
\label{eq:profit_presum}
\mathbb{E}[\pi(s)]
=
F_m(s)\big[(1+\phi)\,b_{h,L}(s)-\phi\big]
+
\bar F_m(s)\big[(1+\phi)\,b_{v,L}(s)-\phi\big].
\end{equation}
Collecting the $(1+\phi)$ terms and using $F_m(s)+\bar F_m(s)=1$ in the no-push case yields the \emph{profit--bias
identity}:
\begin{equation}
\label{eq:profit_bias_identity}
\boxed{\;
\mathbb{E}[\pi(s)]
=
(1+\phi)\,Q(s)-\phi,
\;}
\end{equation}
where:
\begin{equation}
\label{eq:Q_def_methods}
Q(s)\equiv b_{h,L}(s)\Pr(M<s)+b_{v,L}(s)\Pr(M>s)
\end{equation}
is the expected fraction of handle wagered on the losing side. Because $b_{h,L}(s)$ and $b_{v,L}(s)$ are
conditional losing-side shares, $Q(s)$ admits arbitrary dependence between betting and outcome, and the
book's expected profit is affine and strictly increasing in it.

The profit-bias identity generalizes the bookmaker-profit expression of \citet{levitt2004gambling}. Namely, when bet shares are independent of the outcome, the two coincide. However, whenever public betting covaries with the realized outcome, the expressions deviate. Levitt's book profit underestimates the true value whenever the public allocates a larger stake on the eventual losing side, while overestimating when the stake favors the winning side. The mathematical relationship between the profit-bias identity and Levitt's bookmaker profit is detailed in \hyperref[sec:si]{SI}, Sec.~S2.   



Independence of bet shares from the outcome is the pair of conditions: $b_{h,L}(s)=b_h(s)$ and
$b_{v,L}(s)=b_v(s)$, where $b_h(s)\equiv\mathbb{E}[B_h(s)]$ and $b_v(s)\equiv1-b_h(s)$ are the unconditional shares.
Denoting the winning-side home share by $b_{h,W}(s)\equiv\mathbb{E}[B_h(s)\mid M>s]$, the unconditional share is a convex combination of the two: $b_h(s)=F_m(s)\,b_{h,L}(s)+\bar F_m(s)\,b_{h,W}(s)$. Note that $b_{h,L}(s)=b_h(s)$ can only be satisfied when $b_{h,L}(s)=b_{h,W}(s)$. Similarly, $b_{v,L}(s)=b_v(s)$ requires $b_{v,L}(s)=b_{V,W}(s)$. The two conditions therefore collapse to:
\begin{equation}
\label{eq:delta_def}
\boxed{\;\delta(s)\equiv b_{h,L}(s)-b_{h,W}(s)=0,\;}
\end{equation}
which is the discrepancy between the expected home stake share in games the home side loses versus wins. Thus, in order to test whether bet shares covary with realized outcomes, one evaluates $\delta(s)$ against zero. 



The identity \eqref{eq:profit_bias_identity} attributes the book's entire margin to $Q(s)$, but does not afford insight into the sources of its variation. To that end, we write the book's \emph{shading} and the public's \emph{lean} as
signed deviations:
\begin{equation}
\label{eq:theta_lambda_def}
\theta(s)\equiv F_m(s)-\tfrac12,
\qquad
\lambda(s)\equiv b_h(s)-\tfrac12,
\end{equation}
such that $\theta(s)>0$ when the spread exceeds the median outcome and $\lambda(s)>0$
when public betting favors the home side. This allows us to decompose the losing-side share into three distinct and interpretable components.

\begin{proposition}[Anatomy of the book's margin]
\label{prop:margin_anatomy}
Under a symmetric vig $\phi$ and no pushes:
\begin{equation}
\label{eq:margin_anatomy}
Q(s)=\tfrac12+2\,\theta(s)\,\lambda(s)+2F_m(s)\bar F_m(s)\,\delta(s),
\end{equation}
and therefore:
\begin{equation}
\label{eq:margin_anatomy_profit}
\mathbb{E}[\pi(s)]=\underbrace{\frac{1-\phi}{2}}_{\text{hold}}
\;+\;2(1+\phi)\Big[\underbrace{\theta(s)\,\lambda(s)}_{\text{shading}\,\times\,\text{lean}}
\;+\;\underbrace{F_m(s)\bar F_m(s)\,\delta(s)}_{\text{outcome covariance}}\Big].
\end{equation}
The book's expected margin exceeds its hold if and only if the bracket is strictly positive.
\end{proposition}

\noindent The general expression allowing asymmetric vigs and pushes is developed in the \hyperref[sec:si]{SI}
Sec.~S1. Two consequences follow from (\ref{eq:margin_anatomy_profit}). 

\begin{corollary}[Public lean only increases book profit if matched with a shaded line.]
\label{cor:lean_needs_shading}
If the book prices at the median, then $F_m(s)=\bar F_m(s)=\tfrac12$ and $\theta(s)=0$, yielding:
\[
Q(s)=\tfrac12\big(1+\delta(s)\big),
\qquad
\mathbb{E}[\pi(s)]=\frac{1-\phi}{2}+\frac{1+\phi}{2}\,\delta(s),
\]
regardless of the public's lean $\lambda(s)$. Thus, when pricing accurately, the book exceeds its hold if and only if $\delta(s)>0$.
\end{corollary}

\begin{corollary}[Price shading only increases book profit if paired with a public lean.]
\label{cor:shading_needs_lean}
If bet shares are independent of the outcome, $\delta(s)=0$, then
$\mathbb{E}[\pi(s)]=(1-\phi)/2+2(1+\phi)\,\theta(s)\lambda(s)$, which exceeds the hold if and only if
$\theta(s)$ and $\lambda(s)$ are non-zero and have the same sign. If public betting is balanced, price shading does not increase book profit.
\end{corollary}

By making use of the identity $b_h(s)=F_m(s)\,b_{h,L}(s)+\bar F_m(s)\,b_{h,W}(s)$, one can show that $\mathrm{Cov}(B_h(s),\mathbf{1}\{M<s\}) = F_m(s)\bar F_m(s)\,\delta(s)$, meaning that the third term in \ref{eq:margin_anatomy} is equal (up to a constant) to the covariance between the home bet share and whether the home side loses:
\begin{equation}
\label{eq:margin_anatomy_cov}
Q(s)=\tfrac12+2\,\theta(s)\,\lambda(s)+2\,\mathrm{Cov}\big(B_h(s),\mathbf{1}\{M<s\}\big).
\end{equation}
Note that the second term of the right-hand side is a product of marginal expectations while the third term is a covariance. In other words, the book profit's excess factors into a product of marginals and a covariance term, which is standard for a joint expectation of random variables. Although both terms may be viewed as ``public bias'', the two terms have very distinct interpretations: $\theta(s) \lambda(s)$ is a directional bias that is theoretically observable before the match takes place, while $\mathrm{Cov}(B_h(s),\mathbf{1}\{M<s\})$ is outcome-covariant and can only be resolved after. 

It is also instructive is to consider the extreme case of a public that always backs the home side. Here the home bet share is clearly independent of the outcome, and $b_h(s) = b_{h,L}(s) = b_{h,W}(s) = 1$. Despite the fact that $\delta(s)=0$, the book can still profit above its hold if $F_m(s)>\tfrac12$. This is precisely the mechanism presented by \citet{levitt2004gambling} and the second term of \ref{eq:margin_anatomy_cov}. Because Levitt implicitly assumes independence between bet shares and outcomes, 
$\delta(s)=0$ and the product $\theta\lambda$ is the only mechanism by which the book can profit above its hold in his account.

%% file: sections/belief.tex
\subsection*{The profit driver as a misclassification probability}
Note that the profit-bias identity \eqref{eq:profit_bias_identity} makes no assumptions about how the public
chooses sides beyond the existence of the conditional expectations in \eqref{eq:b_def_general}. All market- and
behavior-specific structure enters through $Q(s)$, which depends on the joint distribution of $(B_h(s),M)$. Here we employ a simple public-belief model and show that it leads to a decision-theoretic interpretation of $Q(s)$.   



Consider the point spread setting where $M$ models the difference between home and away score, $s$ is the sportsbook's posted spread, and assume that the public forms a noisy, biased belief $Y$ about the (latent) expected outcome $X$:
\begin{equation}
\label{eq:belief_model_main}
Y = X + \varepsilon + V, \qquad\qquad M = X + U,
\end{equation}
where $\varepsilon$ is the public's systematic bias, $V$ is belief noise, and $U$ captures noise on the outcome that is assumed to be independent of $V$. We further assume that the public backs the home side whenever $Y>s$, such that the home team's bet share is captured by a representative agent: $B_h(s)=\mathbf{1}\{Y>s\}$. In other words, the public acts on the aggregate signal $Y$ and places its stake on one side -- $B_h(s)\in\{0,1\}$. With this selection rule, the conditional losing-side shares $b_{h,L}(s)$ and $b_{v,L}(s)$ of \eqref{eq:b_def_general} become conditional misclassification probabilities, and $Q(s)$ takes on a decision-theoretic interpretation. 


\begin{proposition}[Profit driver as a misclassification probability]
\label{prop:misclassification}
Under the public-belief model with selection rule $B_h(s)=\mathbf{1}\{Y>s\}$, the profit driver of the profit--bias identity satisfies:
\[
Q(s) = \Pr\big(\mathrm{public\ backs\ the\ losing\ side}\big).
\]
$Q(s)$ is thus the misclassification probability of the binary predictor $\mathbf{1}\{Y>s\}$ for the event $M>s$. As a result, the sportsbook's expected profit is an affine, strictly increasing function of the public's classification error in predicting the sign of $M-s$.
\end{proposition}

\noindent Consequently, the book profits to the extent that the public's belief signal fails to classify the outcome. The proof of Proposition \ref{prop:misclassification}, as well as a closed-form evaluation of $Q(s)$ in terms of the belief and outcome distributions, are given in the \hyperref[sec:si]{SI} (Sec.~S4).

%% file: sections/goldilocks.tex
\subsection*{Conditions for the Existence of a ``Goldilocks Zone''}
\label{sec:goldilocks}

We ask when a posted spread $s$ is simultaneously attractive to the sportsbook and to a small bettor: that is, when the book earns a profit in expectation \emph{and} an individual bettor (whose wager is negligible relative to the public aggregate) can find at least one side with positive expected value at the posted price. We assume fractional spreads ($\pi_{\mathrm{push}}(s)=0$) and symmetric vigs ($\phi_h=\phi_v=\phi$ with $\phi\in(0,1)$). For notational convenience, we define:
\begin{equation}
\tau \;\equiv\; \frac{\phi}{1+\phi}\;\in\;\big(0,\tfrac12\big),
\label{eq:tau_def}
\end{equation}
such that $\tau$ and $1-\tau=\frac{1}{1+\phi}$ are symmetric about $\tfrac12$.


From (\ref{eq:profit_bias_identity}), the sportsbook's expected profit is positive if and only if:
\begin{equation}
Q(s) > \tau .
\label{eq:book_profit_condition}
\end{equation}
Note from (\ref{eq:Q_def_methods}) that $Q(s)$ is a convex combination of the two conditional losing-side shares, weighted by the outcome CDF:
\begin{equation}
Q(s) \;=\; \big(1-F_m(s)\big)\,b_{v,L}(s) \;+\; F_m(s)\,b_{h,L}(s).
\label{eq:Q_convex}
\end{equation}
Because $F_m(s)\in[0,1]$, it follows that $Q(s)$ is bounded between the two losing bet shares:
\begin{equation}
\boxed{\;\min\big(b_{h,L}(s),b_{v,L}(s)\big)\;\le\;Q(s)\;\le\;\max\big(b_{h,L}(s),b_{v,L}(s)\big).\;}
\label{eq:Q_sandwich}
\end{equation}
Thus, whether the book profits depends on where the threshold $\tau$ is situated relative to $b_{h,L}(s)$ and $b_{v,L}(s)$. 


\begin{proposition}[Book-profit regimes]
\label{prop:goldilocks_brackets}
Fix a spread $s$ and let $b_{\min}(s)=\min(b_{h,L},b_{v,L})$ and $b_{\max}(s)=\max(b_{h,L},b_{v,L})$. The book-profit condition \eqref{eq:book_profit_condition} obeys exactly one of three regimes:
\begin{enumerate}
\item[\textup{(i)}] \textbf{Always} ($b_{\min}(s)>\tau$): $Q(s)>\tau$ for every value of $F_m(s)$.
\item[\textup{(ii)}] \textbf{Never} ($b_{\max}(s)\le\tau$): $Q(s)\le\tau$ for every value of $F_m(s)$.
\item[\textup{(iii)}] \textbf{Threshold} ($b_{\min}(s)\le\tau<b_{\max}(s)$): as $F_m(s)$ varies from $0$ to $1$, $Q(s)$ increases monotonically from $b_{v,L}(s)$ to $b_{h,L}(s)$, crossing $\tau$ only once at: 
\[
F^\ast(s)=\frac{\tau-b_{v,L}(s)}{b_{h,L}(s)-b_{v,L}(s)}\;\in\;[0,1].
\]
The book profits on whichever side of this crossing puts more of the weight on the larger of the two shares:
\[
Q(s)>\tau
\quad\Longleftrightarrow\quad
\begin{cases}
F_m(s)>F^\ast(s), & \text{if } b_{h,L}(s)>b_{v,L}(s),\\[2pt]
F_m(s)<F^\ast(s), & \text{if } b_{h,L}(s)<b_{v,L}(s).
\end{cases}
\]
\end{enumerate}
In particular, $b_{\min}(s)>\tau$ is \emph{sufficient} and $b_{\max}(s)>\tau$ is \emph{necessary} for the book to profit at $s$.
\end{proposition}

\noindent The proof of Proposition~\ref{prop:goldilocks_brackets} is in the \hyperref[sec:si]{SI} (Sec.~S5). The key idea is that book profitability is governed by the values of the conditional losing-side shares relative to the threshold $\tau$: the book always profits whenever the public's losing-side stake exceeds $\tau$ on \emph{both} outcomes ($b_{\min}>\tau$), while never profiting when it fails to exceed $\tau$ on \emph{either} ($b_{\max}\le\tau$). In the remaining case, book profitability is determined by the value of $F_m$ relative to $F^{\ast}$ at the posted spread $s$.

We define the ``Goldilocks Zone'' as the set of spreads at which the book \emph{and} a small bettor can profit:
\begin{equation}
\mathcal{G} \;\equiv\; \big\{s: Q(s) > \tau\big\}\;\cap\;\mathcal{B},
\qquad
\mathcal{B} \;\equiv\; \big\{s: \text{a negligible bettor has a positive expected value wager at} ~s\big\}.
\label{eq:goldilocks_def}
\end{equation}

The conditions under which a bettor may find a profitable wager have been previously derived \citep{dmochowski2023statistical}. Briefly, a unit wager nets $\phi$ on a win and $-1$ on a loss, such that a bet has positive expected profit only when its win probability exceeds $1/(1+\phi)=1-\tau$. The win probabilities of the home and away sides are given by $\bar F_m(s)$ and $F_m(s)$, respectively, leading to:
\begin{equation}
\mathbb{E}[\text{home bet}]>0 \iff F_m(s) < \tau,
\qquad
\mathbb{E}[\text{visitor bet}]>0 \iff F_m(s) > 1-\tau .
\label{eq:bettor_conditions}
\end{equation}
These conditions amount to $s$ lying in a \emph{tail} of the outcome distribution $F_m$. In other words, the bettor may profit only when the outcome is extreme enough to overcome the vig:
\begin{equation}
\mathcal{B} \;=\; \big\{s: F_m(s) < \tau\big\} \;\cup\; \big\{s: F_m(s) > 1-\tau\big\}.
\label{eq:bettor_set}
\end{equation}

\paragraph{Existence of a Goldilocks Zone.}
Combining the book profitability condtions of Proposition~\ref{prop:goldilocks_brackets} with \eqref{eq:bettor_set} yields the following (distribution-free) conditions on the existence of a non-empty Goldilocks Zone:
\begin{itemize}
\item \textbf{Sufficient.} If, for at least one $s$, the bet shares satisfy $b_{\min}(s)>\tau$ \emph{and} $F_m(s)$ lies in a tail ($F_m(s)<\tau$ or $F_m(s)>1-\tau$), then $s\in\mathcal{G}$ and the Goldilocks Zone is non-empty.
\item \textbf{Necessary.} Any $s\in\mathcal{G}$ must satisfy $b_{\max}(s)>\tau$ \emph{and} $F_m(s)$ in a tail; where $b_{\max}(s)\le\tau$, the book cannot profit and $\mathcal{G}$ is empty regardless of $F_m$.
\end{itemize}
For the threshold regime $b_{\min}(s)\le\tau<b_{\max}(s)$, non-emptiness of $\mathcal{G}$ requires that the tail constraint on $F_m(s)$ and the ``crossover'' constraint ($F_m(s)$ vs.\ $F^\ast(s)$) hold at the \emph{same} spread. Importantly, the quantities $F_m(s)$, $b_{h,L}(s)$, and $b_{v,L}(s)$ all covary with $s$, such that existence is determined by the joint trajectory $s\mapsto\big(F_m(s),b_{h,L}(s),b_{v,L}(s)\big)$, the object that we estimate in the forthcoming empirical section.  


%% file: sections/schematic.tex
\subsection*{Visualizing the profit driver and the Goldilocks Zone}

To visualize the theoretical results, we modeled the outcome
distribution $F_m(s)$ and the public's home bet share $b_h(s)$ with toy Gaussian
models (see \emph{Methods} for details). In the schematics of
Fig.~\ref{fig:Qprocess}, the horizontal axis is a \emph{candidate} point spread
that the sportsbook may post for a match whose median outcome is taken to be zero
(without loss of generality); moving the spread to the right (left) increases the
handicap on the home (away) team. The vertical axis is the profit driver $Q(s)$ --
the expected share of money on the eventual losing side -- and the book profits
when $Q(s)$ exceeds the vig-defined threshold $\tau$. We examine how
book profit responds to four types of public betting behaviors: a \emph{calibrated}
public whose beliefs are aligned with the outcome distribution (panel a); a
\emph{favorite-leaning} public whose bias toward the stronger team (the side the
sportsbook favors) places a preponderance of money on the favorite (panel b); a
\emph{square} public that allocates a higher bet share to the eventual losing side
(panel c); and a \emph{sharp} public that allocates a higher bet share to the
eventual winning side (panel d). It is important to note that the first two behaviors are independent
of the realized outcome, whereas the last two convey a statistical dependence
between betting and outcome.

A calibrated public confines the book's profitable region (i.e., the region of the
curve where $Q(s)$ exceeds $\tau$) to an interior band around the median ($|s|<\epsilon$;
Fig.~\ref{fig:Qprocess}a). Here the book profits \emph{only by pricing
accurately}: its profitability requires that the posted spread be near the true median. The
corresponding Goldilocks regions (green bands) are situated on either side of the median, 
where the posted spread differs from the median by a minimum amount without exceeding $\epsilon$.

A favorite-leaning public -- one that backs the favorite in excess of its
fair probability -- produces a qualitative change in book profit
(Fig.~\ref{fig:Qprocess}b): the book now profits at the tails of the
candidate spread where the public's bias is largest. Note that here the profit curve opens upward (U-shaped). This arises without dependence between betting and
realized outcome: it is a structural phenomenon, rooted in the
the public's tendency to overvalue favorites. The Goldilocks
region widens correspondingly, extending across the tails.

Introducing a dependence between betting and outcome produces
vertical shifts of the profit curve; we model it with a constant $\delta$
between the conditional bet shares (solid line, home loses; dashed line, home
wins). A \emph{square} public, whose money skews toward the eventual loser
(Fig.~\ref{fig:Qprocess}c), quantitatively increases the magnitude of the book's profit without greatly changing the extent of the Goldilocks Zone.
Conversely, a \emph{sharp} public reduces book profit across the range of spreads,
and in the schematic shown drives the book to a loss over a substantial interior
region (Fig.~\ref{fig:Qprocess}d); the Goldilocks region is correspondingly
limited to the extreme tails of the candidate spread.

Taken together, these panels relate book profitability to three mechanisms: (i) accurate forecasting of the outcome distribution; (ii) structural biases in public betting, such as over-valuing favorites; and (iii) any conditional dependencies between betting and outcome. The empirical section below attempts to shed light onto which of these mechanisms may manifest in real betting markets. 

\begin{figure}[t]
    \centering
    \includegraphics[width=\linewidth]{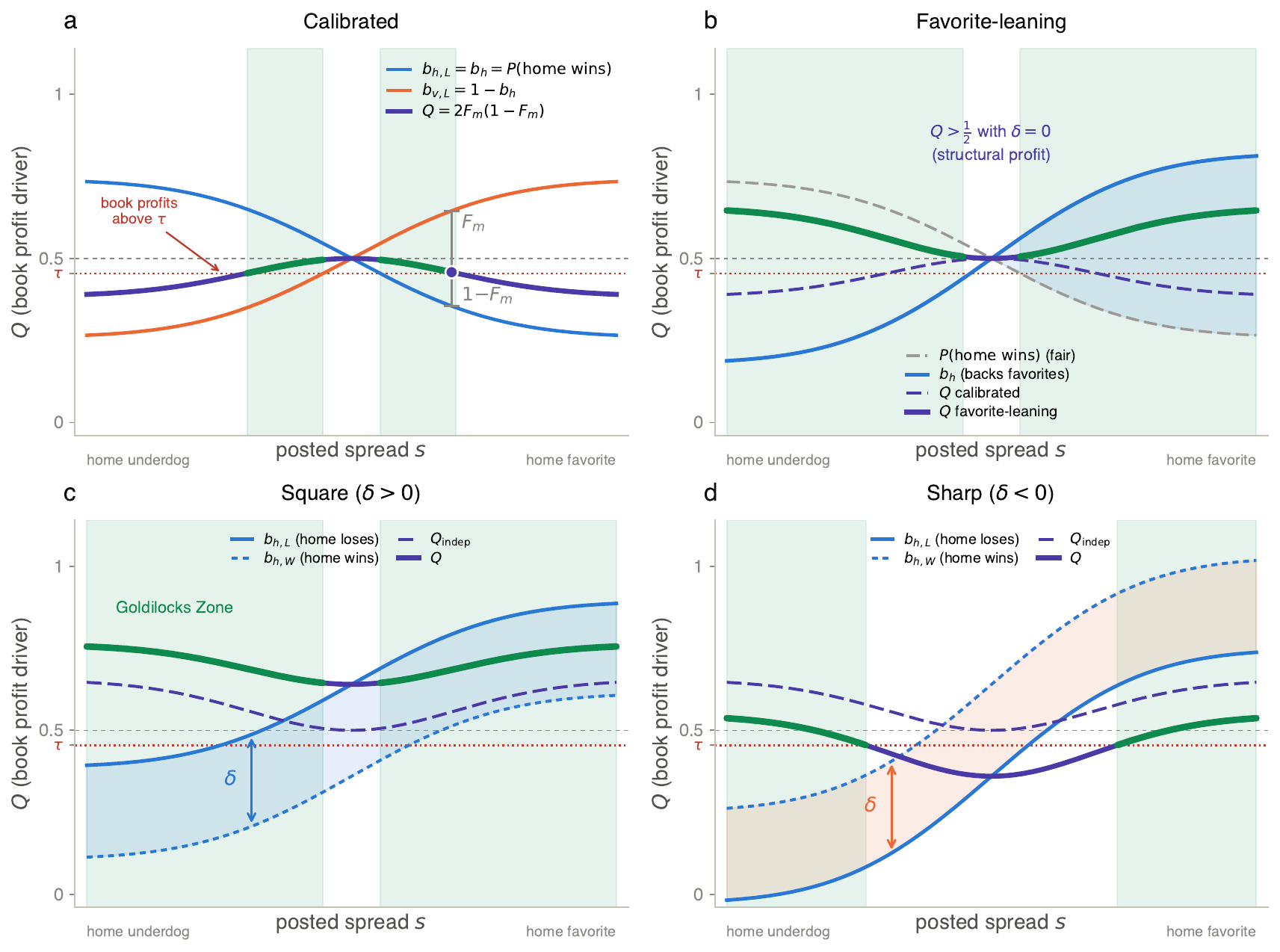}
    \caption{\textbf{The sportsbook profit driver $Q(s)$ across four public
    behaviours}. In every panel
    $Q(s)=F_m(s)\,b_{h,L}(s)+\bar F_m(s)\,b_{v,L}(s)$ is the $F_m$-weighted average
    of the home losing-side share $b_{h,L}=\mathbb{E}[B_h\mid M<s]$ and the visitor
    losing-side share $b_{v,L}=\mathbb{E}[B_v\mid M>s]$; the dotted line marks the
    threshold $\tau$, above which the book profits. \textbf{(a)}~A \emph{calibrated} public bets
    each side at its actual win probability. The bracketed connector shows $Q$
    dividing the gap between the two shares at a ratio of $F_m:(1-F_m)$. The inverted-$U$
    curve is upper bounded by $\tfrac12$ and falls below the profit threshold at the
    tails of the posted spread. \textbf{(b)}~A \emph{favorite-leaning} public backs
    favorites in excess of their fair cover probability. Even with the betting
    fully independent of outcome ($\delta=0$), $Q$ now opens up and is markedly increased at the tails of the candidate
    spread. \textbf{(c)}~A \emph{square} betting public systematically allocates a
    larger bet share on the eventual losing side ($\delta=b_{h,L}-b_{h,W}>0$),
    increasing book profit above a purely favorite-leaning public. \textbf{(d)}~A
    \emph{sharp} betting public allocates a greater bet share on the eventual winning
    outcome ($\delta<0$), driving the book profit towards the threshold $\tau$. In all
    panels, the \emph{Goldilocks} region (green) marks the set of candidate spreads
    where both an individual bettor and the book profit. For a calibrated public (a), it is limited to two interior bands; a
    favorite-leaning public (b) yields $Q$ above $\tau$ across the distribution tails
    and thus widens the Goldilocks Zone. Square
    dependence (c) increases the magnitude of the book's profit without altering the
    extent of the Goldilocks Zone; with a sharp betting public (d), the zone is relegated to more extreme spreads. All curves are schematics generated with
    Gaussian models for both outcome and public belief.}
    \label{fig:Qprocess}
\end{figure}

%% file: sections/empirical.tex
\input{sections/empirical/simpson}
\input{sections/empirical/movement}
\input{sections/empirical/lineshare}
\input{sections/empirical/accuracy}

%% file: sections/empirical/simpson.tex
\subsection*{Empirical results}

To investigate the relationship between sportsbook profit, public betting, and game
outcomes, we collected Major League Baseball (MLB) market prices and bet shares at a major North American sportsbook over an eleven-week duration. The specific market analyzed here, termed the ``run line,'' fixes the favorite's handicap at $1.5$ runs and instead expresses the teams' relative strength through the payouts $\phi_h$ and $\phi_v$. Note that this is unlike a standard point spread market where payouts on the two sides are approximately equal. The bettor's task is thus to predict whether the home team will win by more than (or
lose by less than) $1.5$ runs at a price reflecting the perceived difference in team strengths. Prior tests of Levitt's model have employed empirical bet shares (sometimes termed betting ``splits'') \citep{paul2007sportsbook,paul2008nba,paul2012nhl,humphreys2013homeunderdog}. Here we add a temporal dimension, observing each line and its associated bet shares from posting to game onset, as in previous high-frequency line-movement studies \citep{simon2024inefficient,gandar1998informed}.

We asked the following questions: (i) Do bet shares covary with realized outcomes? (ii) Does the bet share-outcome covariance evolve across time? (iii) Do bet shares drive prices, or do prices drive bet shares, or both? (iv) What are the contributions of price shading, public lean, and outcome covariance to the book's profit? 


\paragraph{Bet share-outcome dependence in MLB run lines: a Simpson's paradox}
Across $n=1{,}139$ games, the betting public displayed a striking preference for the match favorite:
the home bet share varies sigmoidally with the sportsbook's implied win
probability (Fig.~\ref{fig:composition}a), increasing from below
$25\%$ when the away team is strongly favored to above $80\%$ when the home team is.
In contrast, the sportsbook's implied home win probabilities span a compressed
central range ($0.39$--$0.67$).

We computed the home team's bet share separately for games that the home
team won versus lost against the spread. This revealed an apparent discrepancy: in
games where the home side lost the proposition, the home bet share was $5.8$
percentage points higher ($56.4\%$ versus $50.6\%$; game-level bootstrap
$p=0.0004$; $n_L=584$, $n_W=555$). This finding may be interpreted as evidence
for a square betting public. However, stratifying the analysis by the identity of the favorite (home versus away favorite)
removed the effect: the home bet share showed no significant dependence on the
realized outcome in either stratum, and both point estimates were small
(home favorites, $\widehat\delta=-1.5$ pp,
$p=0.51$, $n=608$; away favorites, $\widehat\delta=+0.6$ pp, $p=0.79$, $n=531$).

The apparent dependence is thus a Simpson's paradox: a correlation in pooled data that is removed upon conditioning on a confounding
variable. Here the confounder is the identity of the favorite: when the sportsbook designates the home team as the favorite, both the home bet share $b_h$ \emph{and} the probability of the home team losing against the spread $F_m(s)$ are increased. The latter is due to the fraction of home team losses being slightly larger than that implied by book prices. We elaborate on the implications of this structural form of book profit in the \emph{Discussion}.

\begin{figure}[t]
    \centering
    \includegraphics[width=\linewidth]{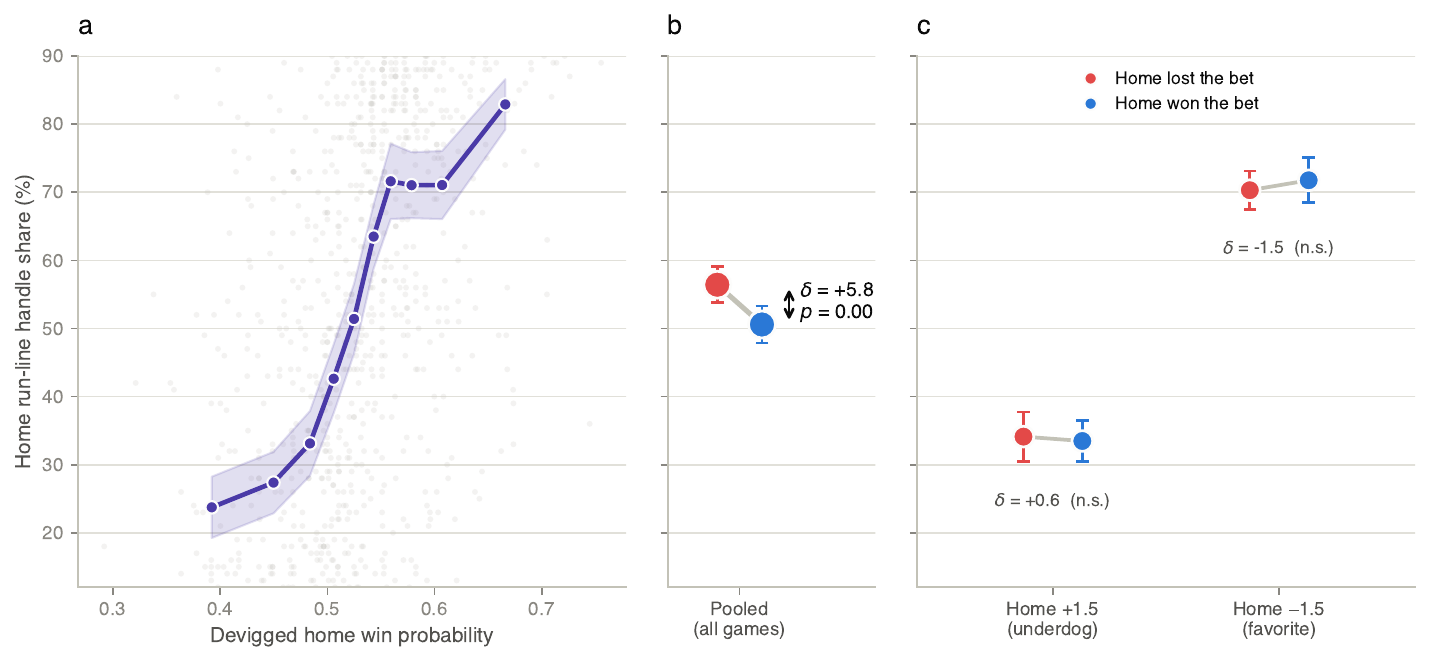}
    \caption{\textbf{A Simpson's paradox explains the apparent bet share-outcome covariance in the MLB run line market.} We collected market prices and bet shares from the MLB run
    line market across $N=1{,}139$ games. \textbf{(a)}~The public shows a striking preference for
    favorites: the home bet share varies sigmoidally with the price-implied home win probability; light grey markers denote individual games, while the curve shows the binned
    mean with $95\%$ confidence band. \textbf{(b)}~Computed over all games, the home bet share appears to covary with the realized outcome: $\widehat\delta=\widehat b_{h,L}-\widehat b_{h,W}=5.8$ points ($p=0.0004$): more money is wagered on the home side in games that the home team eventually lost ($\widehat b_{h,L}$, red) than won ($\widehat b_{h,W}$, blue).
    \textbf{(c)}~However, stratifying the analysis by the identity of the favorite eliminates the effect: $\widehat\delta$ is
    $-1.5$ for home favorites and $+0.6$ for home underdogs (both n.s.). Because a home favorite attracts a majority of the bet share \emph{and} wins the proposition less often than the price-implied win probability, the pooled bet share-outcome covariance is spurious.}
    \label{fig:composition}
\end{figure}

%% file: sections/empirical/movement.tex
\paragraph{Dynamics of book profit and public belief}
Public betting and sportsbook price setting operate as a closed loop, producing fluctuations in both quantities prior to the start of the game. Focusing on the MLB run line market, we collected hourly samples of the point spread $s$, payouts $(\phi_h$, $\phi_v$), and bet shares $(b_h, b_v)$ in the 24 hours leading up to game onset. Within each stratum (home favorite versus home underdog), we calculated the book's profit driver $Q(s)$ as well as the public's outcome dependency $\delta(s)$ in non-overlapping 6 hour windows.

The book's profit driver evolves distinctly across the two strata. For matches with a home underdog, $\widehat Q$ increases from $46.9\%\pm2.6$ at posting to $52.0\%\pm2.1$ at game onset -- the initial ``sharp'' betting is not sustained. In matches with a home favorite, the losing side share is reliably above 50\% throughout ($54.4\%\pm2.4$ to $53.8\%\pm2.0$, mean $\pm$ SEM).

We asked whether the outcome-conditioned bias $\delta$ exhibits significant movement across time, employing a linear mixed model with time and outcome as factors and the home bet share as the dependent variable. A significant interaction between time and outcome indicates that public belief carries information about the eventual outcome that is not constant over the pre-game window -- a dynamic departure from Levitt's implicit independence. For home underdogs, we found a significant decreasing slope, indicating that the home bet share decreases for games in which the home team subsequently wins against the spread ($\widehat b_{h,W}$ falling from $\approx44\%$ a day out to $\approx36\%$ at first pitch; game-clustered bootstrap $p=0.038$, cluster-robust $p=0.038$, $n=256$ games). This finding is suggestive of a square betting pattern that emerges closer to the start of the game. The corresponding test for home favorites was non-significant (game-clustered bootstrap $p=0.85$, $n=288$ games). We note the two strata's slopes do not significantly differ (three-way outcome $\times$ time $\times$ stratum interaction, game-clustered bootstrap $p=0.12$) -- the dynamic departure from independence is specific to home underdogs.

\begin{figure}[t]
    \centering
    \includegraphics[width=\linewidth]{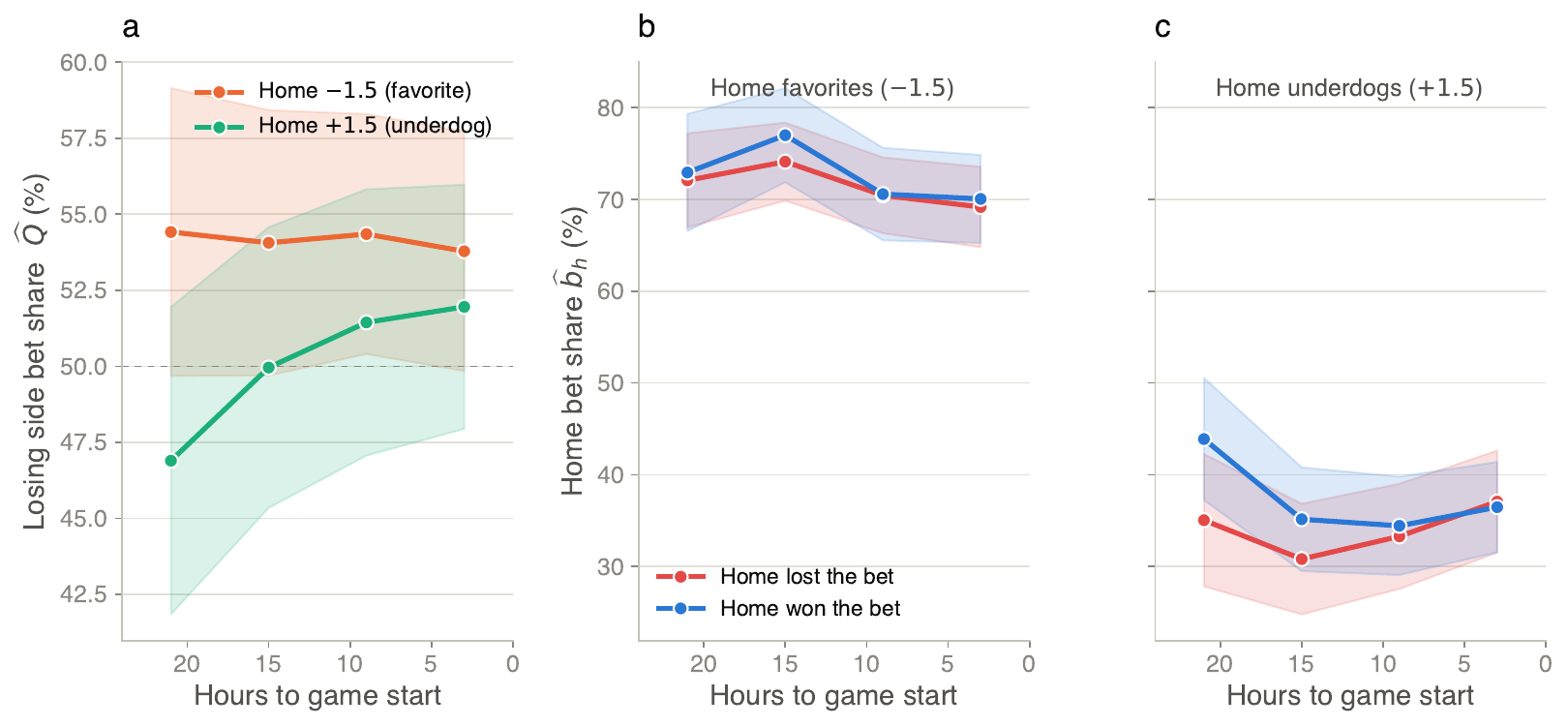}
    \caption{\textbf{Temporal evolution of public bias and book profit during the pre-game window.} 
    We collected hourly samples of public bet shares and prices for the MLB run line market. To gain insight into the temporal dynamics of book profit and public belief, we then estimated the fraction of share on the losing-side $\widehat{Q}$ in non-overlapping six hour windows terminating at game onset. The analysis was stratified into home-favorite ($N=288$) and home-underdog ($N=256$) games.
    \textbf{(a)}~The share of money on the
    eventual losing side $\widehat Q(t)$ by stratum: for home underdogs (green), the losing share increases
    from $46.9\%\pm2.6$ to $52.0\%\pm2.1$, while for home
    favorites (orange) it remains stable and above 50\% ($54.4\%\pm2.4$ to $53.8\%\pm2.0$). \textbf{(b,c)}~The home bet share in games that the
    home side eventually \emph{lost} ($\widehat b_{h,L}$, red) versus \emph{won}
    ($\widehat b_{h,W}$, blue) the closing bet, for home favorites (b) and home
    underdogs (c); the gap between the two curves is the outcome-conditioned bias
    $\widehat\delta(t)$. For games with a home underdog, the home bet share separates by outcome early ($t=-24$ to $t=-18$h: $\widehat\delta = -8.9$ points; more home money on the eventual winner) but converges by the start of the game. The observed change in $\widehat{\delta}$ is significant (linear mixed model with time and outcome as factors, time $\times$ outcome interaction: bootstrap $p=0.038$, $n=256$), suggesting an increase in square betting leading up to the start of the game. Games with a home favorite show little separation at any time ($p=0.85$); the difference between the two strata's slopes is itself not significant ($p=0.12$). }
    \label{fig:movement}
\end{figure}

%% file: sections/empirical/lineshare.tex
\paragraph{Feedback between book pricing and public betting}
The public-belief model treats the posted price as an input to the public's wager
rather than a response to it -- this ordering can be tested directly. We estimated
Jord\`a local projections \citep{jorda2005estimation} of each series on a one-step
innovation in the other (\eqref{eq:lp}; see \emph{Methods}). A local projection
regresses the level of one variable at horizon $h$ on an innovation in the other,
yielding the impulse response horizon by horizon rather than extrapolating it from
a fitted dynamic model. A game fixed effect absorbs the latent state that sets both
the opening price and the public's prior, and a pre-move momentum control blocks the
reverse channel, so that under sequential ignorability the horizon-$h$ coefficient is
the response to the move itself.

On the present dataset, the two directions were markedly asymmetric. A move in the book's line was followed
by a small ``fade'' of the public: a one-percentage-point increase in the devigged home
win probability lowered the home bet share by roughly $0.2$ points, significant at
horizons of one to five hours and peaking at $h=4$ ($\widehat\beta=-0.20$,
$t=-3.11$, $p=0.002$, $n=1090$ games; Fig.~\ref{fig:lp}a). As expected, when wagering on the home side
becomes more expensive, the public's tendency to do so lowers. The
converse channel, however, was absent: the line's response to a one-point move in the home
share was flat at zero and insignificant at every horizon
($|\widehat\beta|\le0.010$ points, all $p>0.27$, $n=1103$ games;
Fig.~\ref{fig:lp}b). The price thus leads the money rather than the reverse,
matching the ordering assumed by the belief model.

\begin{figure}[t]
    \centering
    \includegraphics[width=\linewidth]{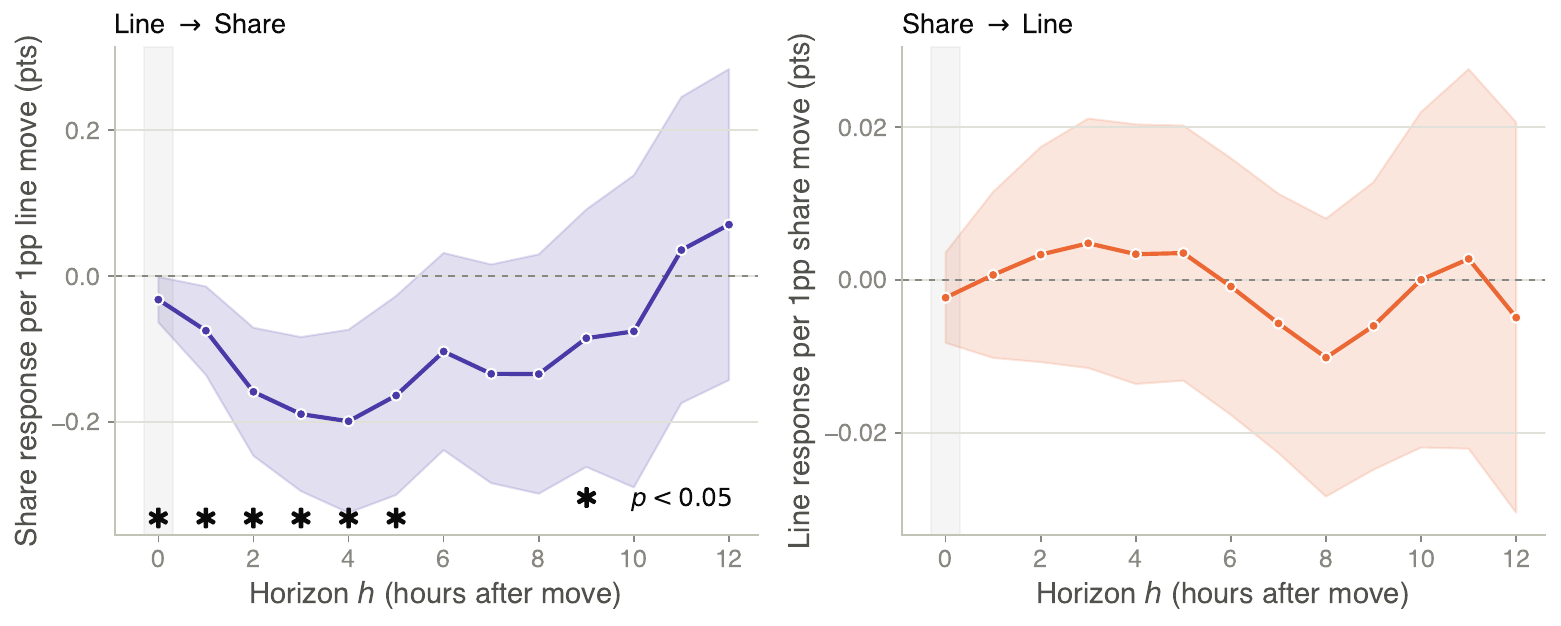}
    \caption{\textbf{Book pricing drives public betting, but not the reverse.} To investigate causal coupling between sportsbook pricing and public betting, 
    we computed Jord\`a local projections of the book price and bet share time series across on an hourly grid ($n=1103$ games; game fixed effect, pre-move momentum
    control). The response of each series at horizon $h$ (number hours after the change in the other signal) is shown with $95\%$ CI, and the shaded strip indicates the
    contemporaneous $h=0$. \textbf{(a)}~\emph{Line $\to$ share:} a $1$-percentage-point
    increase in the devigged home win probability is followed by a small but
    significant reduction of the home bet share ($\widehat\beta\approx-0.20$ points, $t=-3.11$, $p=0.002$, $n=1090$ at the $h=4$ peak; significant at
    $h=1$--$5$; dark-ringed markers).
    \textbf{(b)}~\emph{Share $\to$ line:} in contrast, the line does not significantly respond to a $1$-point change in home bet share
    ($|\widehat\beta|\le0.010$ points, $p>0.27$ at every horizon, $n=1103$). Thus, at least in the present data, the book does
    not reprice to changes in observed bet shares; the price drives the public's belief, not the reverse. }
    \label{fig:lp}
\end{figure}

%% file: sections/empirical/accuracy.tex
\paragraph{The book's profit margin is indistinguishable from its hold.}
Propositions~\ref{prop:margin_anatomy} and~\ref{prop:margin_anatomy_asym} dictate that a book profits beyond its hold
if public betting favors a side that is overpriced, or if bet shares covary negatively with the realized outcome. 
We asked whether either of these two conditions may be present in our dataset. The book earned a mean profit
$5.20\%$ of handle per game ($95\%$~CI $[0.88,9.35]$, $N=1{,}139$; Fig.~\ref{fig:accuracy}a), compared to
the $4.39\%$ hold that even a perfectly balanced book would collect -- on this sample, the book's profit is thus indistinguishable from the hold. Moreover, as is evident in the figure, the per-game profit exhibits very high dispersion (s.d.\ $73.8\%$ of handle). 

Stratifying games by the identity of the favorite, the bet share on the favored side is $+29.6$~pp higher when the home side is the
favorite and $-24.5$~pp higher when the visitor is favored (Fig.~\ref{fig:accuracy}b). Public betting on the favorite thus exceeds the book's implied price by 24--30 percentage points, reflecting a bias that has been previously well documented \citep{levitt2004gambling,paul2007sportsbook,franck2011sentimental}. Note from Corollary~\ref{cor:shading_needs_lean} that a public lean requires a corresponding book shading to yield excess profit. Thus, the empirical finding that the book does not profit above its hold despite the strong favorite-bias suggests that the book's prices are calibrated.

To test this, we estimated the book's mispricing $\hat\theta=p^\ast-\bar{F}_m$ across the price range and contrasted it with the case of a perfectly calibrated book. We found no departure from calibration: the average mispricing is $+0.48$~pp ($95\%$~CI $[-2.42,+3.38]$), and the largest local departure, $5.8$~pp, is within the $9.4$~pp that a calibrated book would produce by chance at this sample size ($p=0.54$). The posted prices are indeed consistent with being calibrated (but see the Discussion for limitations).

For each excess-profit-generating mechanism (shading $\times$ lean, outcome covariance), we asked what the largest level of that mechanism is that remains consistent with the observed data. For shading, we assumed that the book shades in the maximally profitable direction -- the side that the public backs is priced as more likely to cover than it truly is -- and calculated the largest shading whose implied profit margin still falls inside the book profit confidence interval. Because the public's lean is large, (counterfactual) shading would be highly profitable: $1.35$\% of handle per percentage point. Consequently, the upper bound on price shading is tight: at most $3.7$~pp (Fig.~\ref{fig:accuracy}c), meaning that the devigged book price departs from the true cover probability by no more than $0.037$ in the direction of the public lean.

We repeated the analysis for outcome-covariance, calculating the largest $\delta$ compatible with the observed book profit. The bound is again modest: $5.2$~pp (Fig.~\ref{fig:accuracy}d), implying that the public's home bet share differs by at most $5.2$ percentage points between games won versus lost by the home side. 

\begin{figure}[t]
    \centering
    \includegraphics[width=\linewidth]{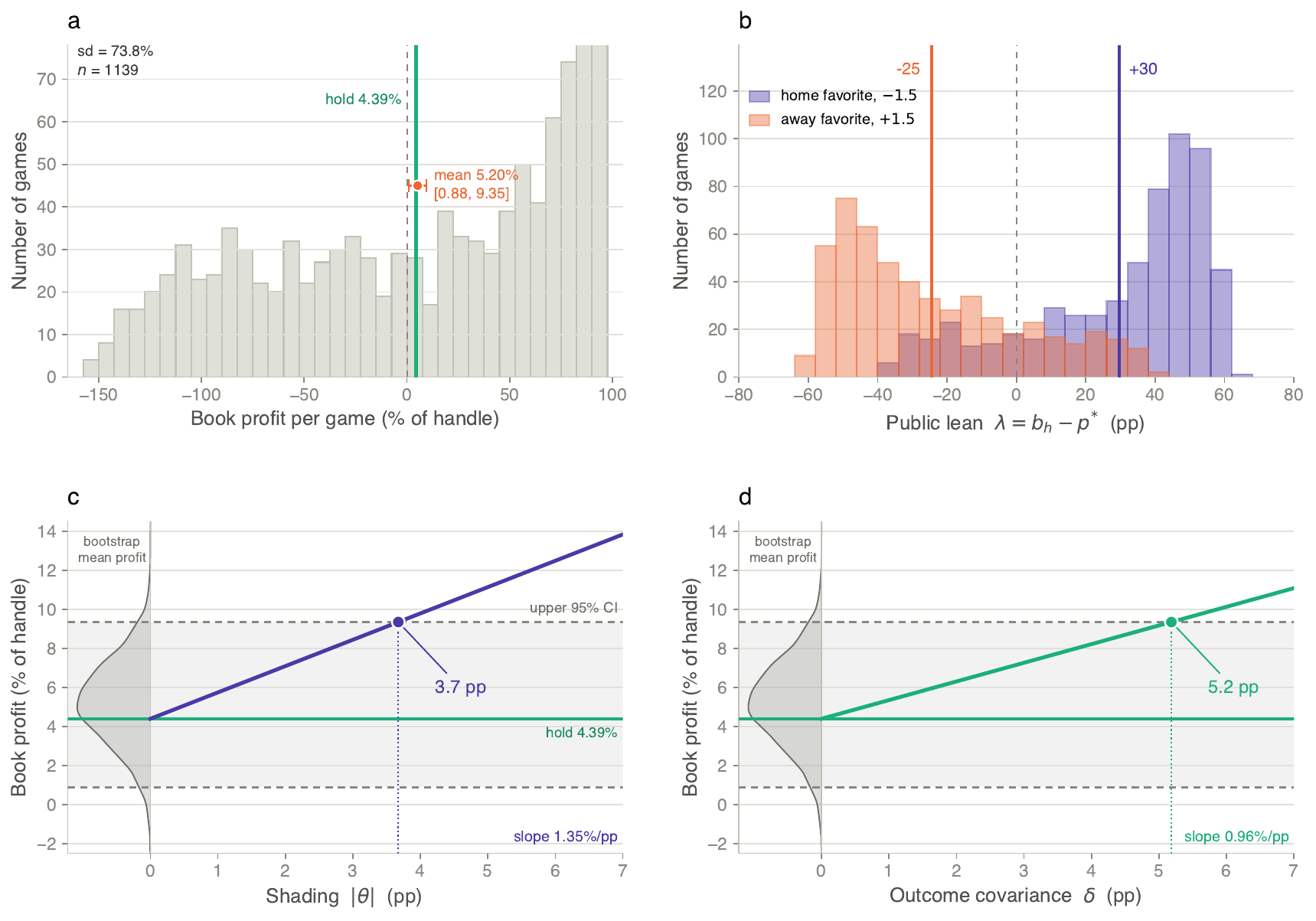}
    \caption{\textbf{The book's profit margin is indistinguishable from its hold, despite a
    large public lean toward favorites.} \textbf{(a)}~Reconstructed book
    margin per game, shown as a percentage of that game's handle ($N=1{,}139$). The mean profit margin of $5.20\%$
    ($95\%$~CI $[0.88,9.35]$, orange) is statistically indistinguishable from the
    $4.39\%$ margin that a perfectly balanced book would collect from the hold alone (green),
    though individual games are widely dispersed (s.d.\ $73.8\%$).
    \textbf{(b)}~The public lean $\lambda=b_h-p^\ast$, split by the identity of the favored side.
    Betting is highly skewed to the favorite: $+29.6$~pp where the home side is the
    favorite, $-24.5$~pp where the visitor is favored. \textbf{(c,d)}~Upper bounds on the two mechanisms that
    the book may exploit to profit beyond the hold. The grey density is the bootstrap distribution
    of the mean profit margin, with its $95\%$ CI interval marked; the sloped line denotes the excess profit margin that the book earns as a function of the level of the corresponding mechanism (either price shading or outcome-covariance). The largest values of shading and outcome covariance that are compatible with the empirical sample are thus marked by the lines' crossing of the upper CI:  $3.7$~pp of shading aligned with the lean, and $5.2$~pp of outcome covariance.}
    \label{fig:accuracy}
\end{figure}

%% file: sections/discussion.tex
\section*{Discussion}

%
%
%
%
%

%
%
We argue that Levitt's influential model of bookmaker profit \citep{levitt2004gambling} assumes that public betting is independent of the realized outcome, and then generalize the model of sportsbook profit to allow betting-outcome dependence. Importantly, this leads to an expression for bookmaker profit that is affine in the public's bet share on the eventual losing side: the profit-bias identity. That central object of the identity may be understood in the framework of decision theory, such that bookmaker profit is synonymous with the public's misclassification error of the outcome relative to the book's line. This places the book's accounting and the bettor's estimation problem on a common footing, since the same quantiles of the outcome distribution that govern one govern the other \citep{dmochowski2023statistical}. There is thus a direct link between how much the book profits and how biased the public's belief of the outcome is -- at the posted spread. A profit-maximizing book seeks to identify the price which maximizes the share of betting on the losing side but weighted by the outcomes of the corresponding loss. This is inherently a three-dimensional optimization problem -- for each candidate point spread, the profit-optimizing book must forecast (i) the share of bets on the losing side if home wins, (ii) the share of bets on the losing side if visitor wins, and (iii) the probability of the home side covering. Critically, (i) and (ii) do \emph{not} generally sum to one: their sum is $1+\delta(s)$, and equals one precisely when the public's money is statistically independent of the outcome. That independence is the scenario treated by Levitt, whose profit expression follows as a special case of the profit-bias identity. Generalizing the model does not, however, overturn his conclusion: that books shade prices to exploit predictable public bias remains consistent with the profit-maximizing view of fixed-odds pricing \citep{kuypers2000information} and with evidence that books do not move lines to equalize the money on the two sides \citep{paul2007sportsbook}. What the present work adds is that the proximate driver of book profit is the public's misclassification error of the outcome relative to the posted line.  

A surprising finding is that the book can profit above its hold even in the absence of statistical dependence between outcomes and betting. The bookmaker profit decomposes into three "channels": (i) its hold, (ii) a structural profit term that is positive when the public leans towards the same side that the book's prices are shaded, and (iii) the covariance between betting and outcomes. Perhaps counterintuitively, we show that public lean and book shading may align even if the public has no information pertaining to the realized outcome. As long as the side with the higher bet share has a lower probability of winning than that implied by its price, the book earns a structural profit in addition to the hold. The lean that drives this channel is itself well documented, and is usually attributed to the loyalties and sentiment of a recreational public rather than to information \citep{avery1999sentiment, franck2011sentimental, stanek2017homebias}. The mechanism is complementary to the classical account of bookmaker margins, in which both the margin and the favorite--longshot bias arise from the threat of privately informed insiders \citep{shin1991optimal, shin1992prices}; here no informed trader is required, only a marginal taste for favorites meeting a handicap that makes that taste costly. That bias takes the opposite sign in parimutuel racetrack betting, where longshots are overbet \citep{thaler1988anomalies, snowberg2010explaining, ottaviani2008favorite}; the favorite preference we observe is its fixed-odds, team-sport counterpart, of a piece with earlier study of the baseball betting market \citep{woodland1994baseball}. Note that the conditions for this second channel of book profit are known before the actual outcome is observed: the collected bet shares and likelihood of each side winning the bet. In order for the book to profit from the third channel, it must elicit a preponderance of money on the side that loses the match (presumably by identifying a price at which this occurs systematically). This third channel cannot be resolved until after the game has concluded. 

Another surprising aspect of our work pertains to our empirical investigation of the MLB run-line market at a major North American sportsbook, where we discovered that a Simpson's paradox can spuriously indicate betting-outcome dependence. When pooling all matches, we observed a significant majority of the public's money on the eventual losing side (Fig.~\ref{fig:composition}b). However, this effect is completely eliminated once stratifying the data by the identity of the favorite: in games with a home underdog, there was no significant difference between the bet share on the losing versus winning side (Fig.~\ref{fig:composition}c). The same finding was observed in matches where the home side was favored (Fig.~\ref{fig:composition}c). Yet when pooling the two strata, a seeming ``square'' public betting pattern emerges. The resolution to this Simpson's paradox is that the identity of the favorite drives both public betting tendencies (because the public prefers favorites) \emph{and} the eventual outcome (because the favorites win less often than implied by the book's prices). The finding of more money on the eventual losing side when aggregating all games is thus a composition artifact. Importantly, this result is a caution to tests of market efficiency and betting-outcome dependence, which are often performed on pooled samples \citep{paul2002totals, paul2011nfl, humphreys2013homeunderdog}. The excess-share term $2F_m\bar F_m\delta$ of \eqref{eq:Q_excess} makes the separation explicit: what a pooled test reads as behavioral dependence is the sum of a within-stratum term and a composition term, and only the former carries the interpretation usually assigned to it. Our findings point to the importance of conditioning such analyses on match attributes. 

Our empirical analysis reconstructed book profit over $1{,}139$ games and found that the resulting profit margin does not vary from that of the book's guaranteed hold (Fig.~\ref{fig:accuracy}a). This is despite very strong evidence for the public's bias towards the match favorite (Fig.~\ref{fig:accuracy}b). From the second channel of book profit \eqref{eq:margin_anatomy_profit}, a margin no larger than the hold despite so pronounced a lean implies that the book's price shading must be small: the data bound it to at most $3.7$~pp in the direction of the lean, and a direct check of the posted prices shows no detectable departure from calibration. We read this as a null on shading rather than as a demonstration of accurate pricing. Expected profit depends on the mispricing only through its product with the lean, so a margin test carries power only against mispricing that aligns with where the money sits: in simulation on our own design, $5$~pp of aligned mispricing is detected $88\%$ of the time, against $6\%$ when the same mispricing is oriented independently of the lean. The direct check is limited in a complementary way, since at this sample size a perfectly calibrated book would itself generate local departures of up to $9.4$~pp by chance. What we bound is therefore Levitt's mechanism specifically; a book could misprice substantially in ways unrelated to the public's lean and leave no trace in its margin. Because the book can only profit from a public lean towards favorites if the prices are also shaded towards favorites (Cor.~\ref{cor:shading_needs_lean}), the absence of detectable shading nullifies the potential additional profit. Interestingly, this scenario finds the book taking on additional risk -- the variability in profit margin is enormously higher than that of the hold itself -- without ``reward''. The expected profit is no larger than the risk-free hold, but comes with a much wider confidence interval.

One limitation of the present study is that the treatment is limited to two-outcome markets, whereas several prominent markets are three-way (e.g. European football home-draw-away). As such, extending our framework to multi-way markets such that the profit-bias relationship can be quantified represents an avenue for future research. Our empirical results are best understood as demonstrations of the proposed theoretical framework rather than inferential claims about the economics of modern-day sportsbooks. We analyzed data from one sport and one market (albeit both major) and the empirical findings serve to demonstrate that key objects such as $Q(s)$, $\delta$ may be estimated from data and that they convey insight into that specific market. The number and size of sports betting markets is vast, and we do not aim to characterize how those objects vary across the spectrum of sports and bet types. Relatedly, public bet-share data are scarce. Here we collected hourly data from a single book over several months, but quantities such as the extent of the Goldilocks Zone require the joint distribution of bet share and outcome over many games and remain hard to estimate. Prediction markets -- whose order books are an analogue of the bet share -- may generate richer data \citep{wolfers2004prediction, snowberg2013prediction}, but it is unclear how neatly this may map to sportsbook markets. A further caveat is that the book's margin is reconstructed rather than observed: the splits record shares rather than dollars, so each game enters per unit of its own handle, and the aggregate depends on a volume distribution we do not see. Finally, with the exception of the local-projection analysis (Fig.~\ref{fig:lp}), our treatment is non-causal. Whether pre-game movement of the kind we observe reflects the arrival of informed money is the subject of an extensive line-movement literature \citep{gandar1998informed, krieger2013price, simon2024inefficient, croxson2014information}.
On the other hand, betting markets likely exhibit a deep causal structure linking information arrival, book prices, and public betting tendencies; as such, a causal account of how book prices and public beliefs evolve is an important direction for future work.

 The profit-bias identity naturally organizes betting markets by the properties of the conditional losing shares. In particular, we define the Goldilocks Zone as the set of prices at which the sportsbook \emph{and} a negligible bettor are simultaneously positive in expectation. This requires that the public's losing-side share exceeds the book's break-even threshold while the outcome distribution is sufficiently removed from the median. 
 A Goldilocks Zone represents an inefficiency that still allows the book to profit, and thus connects the present work to the long-standing debate over betting market efficiency \citep{sauer1998economics, gray1997testing, vandenbruaene2022efficient}. Whether such a zone can be exploited in practice depends on the same constraints that limit arbitrage in financial markets \citep{shleifer1997limits, moskowitz2021asset}: betting limits, transaction costs, and the difficulty of consistently identifying the profitable side.

The view that sportsbooks profit by exploiting a predictably biased public has organized the literature for two decades. Legalized online betting is now expanding that market and shifting the composition of the betting population. It remains to be seen whether the prevalent view outlasts this change; the profit-bias identity provides the means to answer that question empirically.

%% file: sections/methods.tex
\section*{Methods}

\subsection*{Construction of schematic figure (Fig.~\ref{fig:Qprocess})}
Figure~\ref{fig:Qprocess} is an illustrative schematic; as such, its parameters are chosen
for clarity rather real-world plausibility. Outcomes follow a Gaussian model where the median outcome (home score - away score) is assumed to be 0 (center of the horizontal axis) and the horizontal axis represents the sportsbook's posted point spread, which may be viewed as their estimate of the outcome. We define $z(s)=2\Phi(s/2)-1,~z \in(-1,1)$, where $\Phi$ is the standard normal CDF, to capture the strength of the implied favorite.  The outcome distribution is then modeled as $F_m(s)=\tfrac12+0.24\,z(s)$, such that favorites (larger $|s|$) win the proposition less than half the time. The sportsbook's break-even threshold is set to a realistic vig, $\phi=100/120$ ($-120$
American odds), yielding $\tau=\phi/(1+\phi)\approx0.45$. The four public behaviours
differ only in the home bet share $b_h(s)$ and the outcome-covariance $\delta(s)$:
\textbf{(a)}~\emph{calibrated}, $b_h=\bar F_m$ (each side backed at its actual win
probability), so that $Q=2F_m(1-F_m)$; \textbf{(b)}~\emph{favorite-leaning},
$b_h=\tfrac12+0.32\,z$ (favorites backed beyond their actual win probability) with
$\delta=0$, such that $Q=\tfrac12+2(F_m-\tfrac12)(b_h-\tfrac12)$; and
\textbf{(c)}~\emph{square}-\textbf{(d)}~\emph{sharp}, which add a constant $\delta=\pm0.28$ such that $Q=\tfrac12+2(F_m-\tfrac12)(b_h-\tfrac12)+2F_m\bar F_m\,\delta$. We clip $\delta$ to the range
$\min\!\big(b_h/F_m,(1-b_h)/\bar F_m\big)$ to ensure that the conditional shares
$b_{h,L}=b_h+\bar F_m\,\delta$ and $b_{h,W}=b_h-F_m\,\delta$ remain in $[0,1]$. The shaded Goldilocks region shows the set of candidate spreads where the book
profits ($Q>\tau$) while still allowing a negligible bettor to find a profitable wager ($F_m<\tau$ or $F_m>1-\tau$).

\subsection*{Data collection}

We obtain public betting data from the DraftKings Network betting-splits website
\citep{dknetwork_splits}, which reports, for each upcoming event and market, the
percentage of handle and tickets wagered on each
side. The splits are collected each hour, yielding a series of samples from posting to game start: we record the event, market, selection, posted odds, and the handle and ticket shares. We focus on the two-sided market to which \eqref{eq:profit_bias_identity} most
directly applies: the baseball \emph{run line}. The run line is a variant of conventional point spread but with an important modification: whereas a
standard point spread varies the handicap and holds the payouts ($\phi_h$, $\phi_v$) roughly symmetric
(e.g. $\phi_h=\phi_v=0.91$), the run line fixes the handicap at $\pm1.5$ and
expresses the teams' relative strength through (generally asymmetric) payouts. As an example, a home favorite with a $-1.5$ handicap may pay $\phi_h=1.5$ with the corresponding visiting underdog ($+1.5$) paying $\phi_v=0.56$. The book thus moves the price (i.e., not the magnitude of the handicap) and it is the combination of the payouts and identity of the handicapped side that conveys the favored side. 



We obtain final scores from the public ESPN scoreboard API \citep{espn_scoreboard}. For each event, we
determine whether it won, lost, or pushed against the line \emph{in effect at that
snapshot}. A run-line wager wins when the sum of the selected team's score and its $\pm 1.5$ handicap
exceeds that of its opponent's. The dataset analysed here spans games from 2026-05-31 through 2026-08-28


\subsection*{Empirical testing of the bet share-outcome dependence (Fig.~\ref{fig:composition})}
To empirically probe the statistical dependence between public betting and realized outcomes, we analyzed MLB run-line bet shares measured at game onset. Specifically, we calculated the sample estimate $\widehat{\delta} = \widehat b_{h,L}$ and $\widehat b_{h,W}$ -- the difference in home bet share between matches where the home team won versus lost against the spread. The quantity $\widehat{\delta}$ was first computed over all games, and then on two individual strata -- the set of games with a home favorite, and those with a home underdog. The rationale for stratifying the data (i.e., conditioning on the identity of the favorite) was to control for hidden variables that produce spurious dependencies between bet share and realized outcome.

Figure~\ref{fig:composition}(a) plots the home bet share
against the home win probability, which we computed from the sportsbook's ``moneyline'' payouts (the moneyline is a bet on which team wins the contest, regardless of margin). The reason for this choice, which affects only the display item, was to allow for sorting the matches along the ``strength of favorite'' dimension. For home and visiting moneyline payouts of $\phi_h$ and $\phi_v$, respectively, we remove the book's vig (``devig'') by normalizing the corresponding implied win probabilities $1/(1+\phi_h)$ and $1/(1+\phi_v)$ to sum to one:
\begin{equation}
\label{eq:devig}
\mathrm{dv}(\phi_h,\phi_v) \;=\; \frac{1/(1+\phi_h)}{\,1/(1+\phi_h)+1/(1+\phi_v)\,}.
\end{equation}
We use the resulting devigged home-win probability $\mathrm{dv}(\phi_h,\phi_v)$ to order games by favorite strength. 
In the figure, games are grouped into deciles, markers are decile means, and the band denotes bootstrapped
$95\%$ CI. Panels (b) and (c) report $\widehat b_{h,L}$ (home lost the bet) and
$\widehat b_{h,W}$ (home won), pooled and within strata respectively,
along with $\widehat\delta=\widehat b_{h,L}-\widehat b_{h,W}$ and its two-sided
bootstrapped $p$-value.

\subsection*{Dynamics of book profit and bet share-outcome dependence (Fig.~\ref{fig:movement})}
To study the temporal dynamics of book profit and associated bet share-outcome dependence, we measured the losing-side share $\widehat{Q}$ and the home bet share $B_{g,t}$ across four, non-overlapping $6$-hour windows terminating at game onset. Only games with at least one sample in each window were included. Within each game, we averaged all samples from the same window. We then fit, separately for
home favorites and home underdogs, the following mixed model:
\begin{equation}
\label{eq:lmm_move}
B_{g,t}=\beta_0+\beta_h\,h_{g,t}+\beta_W\,W_g+\beta_{hW}\,(h_{g,t}\times W_g)+u_g+\varepsilon_{g,t},
\qquad u_g\sim\mathcal N(0,\sigma_u^2),
\end{equation}
where $B_{g,t}$ is the observed home bet share at sample $t$ of game $g$, $h_{g,t}$ is the number of hours before onset of game $g$,
$W_g\in\{0,1\}$ is a binary indicator of whether the home side won the
closing bet, and $u_g$ is a per-game random intercept. Since
$\delta(h)=b_{h,L}(h)-b_{h,W}(h)=-\beta_W-\beta_{hW}\,h$, statistical significance of the closing share is
assessed with the regression coefficient $\beta_W$ (at $h=0$). The \emph{movement} of the dependence is assessed by the
outcome$\times$time interaction $\beta_{hW}$. The reported p-values were obtained with a game-level cluster
bootstrap (resampling whole games with replacement and refitting the model). Figure~\ref{fig:movement} displays, for each stratum, the losing-side share $\widehat Q(h)$ and the conditional bet shares
$\widehat b_{h,L}(h)$ and $\widehat b_{h,W}(h)$. The shares are first averaged within each game (across samples in a $6$-hour bin) and then across games. Shading denotes across-game SEM.

\subsection*{Line--share feedback (Fig.~\ref{fig:lp})}
We investigate the direction of the line--share feedback: whether the book's price
drives the public's share, the reverse, or both. Since the run-line handicap is
fixed to $\pm1.5$, line movement is reflected in the payouts. We represent the momentary line by
$\ell_{g,t}=\mathrm{dv}(\phi_{h,g,t},\phi_{v,g,t})$, obtained by applying the devig transform \eqref{eq:devig}
to the \emph{run-line} payouts at sample $t$. In other words, we operationalize ``line'' to the devigged probability that the
home side covers $\pm1.5$. We then fit Jord\`a local
projections in both directions, for horizons $h=0,\dots,12$ hours:
\begin{align}
\label{eq:lp}
B_{g,t+h}-B_{g,t-1} &= \beta^{\ell\to B}_{h}\,\Delta\ell_{g,t}
  + \gamma_h\,\Delta B_{g,t-1} + a_g + \varepsilon_{g,t},\\
\ell_{g,t+h}-\ell_{g,t-1} &= \beta^{B\to\ell}_{h}\,\Delta B_{g,t}
  + \gamma'_h\,\Delta\ell_{g,t-1} + a'_g + \varepsilon'_{g,t},\nonumber
\end{align}
where $\Delta\ell_{g,t}=\ell_{g,t}-\ell_{g,t-1}$ and $\Delta B_{g,t}=B_{g,t}-B_{g,t-1}$
are the one-step line and share deviations, the outcome is measured from the
pre-move level $t-1$, the game fixed effect $a_g$ absorbs the latent state, and the
lagged one-step move of the responding series (i.e., $\Delta B_{g,t-1}$ and $\Delta\ell_{g,t-1}$, which we term the \emph{pre-move
momentum}) prevents the book from ``chasing'' the money already wagered. Standard errors are clustered by game. Under sequential ignorability (i.e., the current
move is as-good-as-random given the game state and the money so far), $\beta_h$ is
the impulse response of one series to a move in the other. We denoise the integer-rounded bet shares with a $3$-hour moving
average. 

\paragraph{The book's profit margin is indistinguishable from its hold. (Fig.~\ref{fig:accuracy}).}
For each game (indexed here by $i$), we reconstruct the book's profit margin per unit of that game's handle
from three observables: the closing home bet share $b_{h,i}$, the home and away payouts $\phi_{h,i},\phi_{v,i}$, and the realized outcome
$C_i\in\{0,1\}$ representing whether the home side covered:
\begin{equation}
\label{eq:pi_reconstructed}
\pi_i=C_i\big[(1-b_{h,i})-\phi_{h,i}b_{h,i}\big]+(1-C_i)\big[b_{h,i}-\phi_{v,i}(1-b_{h,i})\big].
\end{equation}
Here $\pi_i$ is a profit margin per unit of betting volume, which we have assumed to be equivalent for all games in the dataset.  Denoting the raw implied win probability with $a=1/(1+\phi)$, the corresponding \emph{devigged} win probability is calculated as $p^\ast_i=a_{h,i}/(a_{h,i}+a_{v,i})$. We then tabulate 
$D_i=2+\phi_{h,i}+\phi_{v,i}$, the per-game hold $h_i=(1-\phi_{h,i}\phi_{v,i})/D_i$, and the
public lean $\lambda_i=b_{h,i}-p^\ast_i$. Evaluating Prop.~\ref{prop:margin_anatomy_asym} at the
realized outcome (rather than the cover probability) yields the following expression for per-game profit:
\begin{equation}
\label{eq:pi_identity}
\pi_i=h_i+D_i\,\lambda_i\,(p^\ast_i-C_i).
\end{equation}
Confidence intervals for the mean across-game profit margin $\bar\pi = \frac{1}{N}\sum_i \pi_i$ are percentile intervals from a
nonparametric bootstrap over games ($4{,}000$ resamples).

To upper bound the level of price shading or outcome-covariance that is consistent with the observed data, we calculated the largest levels that yield profit margins within the observed confidence intervals. For price shading that is maximally aligned with the public lean,
$\theta = c\,\mathrm{sgn}(\lambda)$, where $c>0$ is the magnitude of price shading in probability units. Assuming $\delta=0$, the expected profit margin follows as $h + c\,\mathbb{E}[D|\lambda|]$. For a constant outcome covariance $\delta$ with
$\theta=0$ and $\bar F_m=p^\ast$, the profit margin is $h+\delta\,\mathbb{E}[D\,p^\ast(1-p^\ast)]$.
In both cases, the profit margins are linear in the level and have a slope fixed by the observed lean and price
distribution; the largest compatible value is thus given by where line crosses the upper limit of
the bootstrap interval. 

To assess overall calibration of book prices, we pool games and calculate the mean deviation between devigged price and realized cover rate, employing a binomial
standard error. To assess calibration across the price
range, we employed Gaussian kernel regression: on a grid of $200$ points spanning the observed range of
$p^\ast$, we estimate $\Pr(\text{cover}\mid p^\ast)$ as a Gaussian-weighted average of the
realized cover indicators $C_i$, each game weighted by its distance in price from that point
(Nadaraya--Watson regression, bandwidth $0.05$). We then form the miscalibration
curve as $\hat\theta(p^\ast)=p^\ast-\widehat{\Pr}(\text{cover}\mid p^\ast)$, which we summarize
with its largest absolute value $\sup|\hat\theta|$. 

To generate a null distribution for $\sup|\hat\theta|$, we hold the observed prices fixed and draw
counterfactual outcomes according to the book's implied probabilities: $C_i\sim\mathrm{Bernoulli}(p^\ast_i)$. We then recompute the supremum over $4,000$ draws -- the $p$-value is the fraction of simulated suprema that reach the observed one.

%% file: sections/supplement.tex
%

\clearpage
\appendix
\renewcommand{\theequation}{S\arabic{equation}}
\setcounter{equation}{0}
\renewcommand{\thetheorem}{S\arabic{theorem}}
\setcounter{theorem}{0}

\phantomsection
\label{sec:si}
\section*{Supplementary Information}

\subsection*{Notation}
Table~\ref{tab:symbols} defines all symbols used in the paper, where we indicate the section in which each variable first appears.

\begin{table}[htbp]
\centering
\caption{\textbf{Symbols used in the paper.} All bet shares are fractions of the total handle (i.e., normalized to 1).}
\label{tab:symbols}
\small
\begin{tabular}{@{}l p{0.58\textwidth} l@{}}
\toprule
\textbf{Symbol} & \textbf{Meaning} & \textbf{Defined} \\
\midrule
\multicolumn{3}{@{}l}{\emph{Outcome and line}}\\
$M$ & Realized outcome: margin of victory, total, or player statistic & Setup \\
$s$ & Sportsbook's posted line (e.g. point spread) & Setup \\
$F_m(s)$ & $\Pr(M<s)$, the outcome CDF at the line & \eqref{eq:Fm_methods} \\
$\bar F_m(s)$ & $\Pr(M>s)$; the home side wins the proposition & \eqref{eq:Fm_methods} \\
$\pi_{\mathrm{push}}(s)$ & $\Pr(M=s)$, the probability of a ``push' where all bets are cancelled & \eqref{eq:Fm_methods} \\
\addlinespace
\multicolumn{3}{@{}l}{\emph{Bet shares}}\\
$B_h(s),\,B_v(s)$ & Public's random bet share on the home / visitor side, $B_v=1-B_h$ & Setup \\
$b_h(s),\,b_v(s)$ & Unconditional expected bet shares, $b_h=\mathbb{E}[B_h]$ & \eqref{eq:delta_def} \\
$b_{h,L}(s)$ & $\mathbb{E}[B_h\mid M<s]$: home bet share when the home side loses & \eqref{eq:b_def_general} \\
$b_{v,L}(s)$ & $\mathbb{E}[B_v\mid M>s]$: visitor bet share when the visitor loses & \eqref{eq:b_def_general} \\
$b_{h,W}(s)$ & $\mathbb{E}[B_h\mid M>s]$: home bet share when the home side wins & \eqref{eq:delta_def} \\
$b_{\min},\,b_{\max}$ & $\min$ and $\max$ of $b_{h,L}$ and $b_{v,L}$ & Prop.~\ref{prop:goldilocks_brackets} \\
\addlinespace
\multicolumn{3}{@{}l}{\emph{Prices and the vig}}\\
$\phi$ & Net payout per unit wager on a winner (symmetric vig) & \eqref{eq:tau_def} \\
$\phi_h,\,\phi_v$ & Net payouts on the home / visitor side (asymmetric vig) & S1 \\
$\tau$ & $\phi/(1+\phi)$, the book's break-even losing-side share & \eqref{eq:tau_def} \\
$p^\ast$ & Devigged (overround-removed) price of the home side & Prop.~\ref{prop:margin_anatomy_asym} \\
$O$ & Overround: the two raw implied probabilities summed & Prop.~\ref{prop:margin_anatomy_asym} \\
$h$ & Hold: share of handle that a balanced book keeps, $h=1-1/O$ & Prop.~\ref{prop:margin_anatomy_asym} \\
$D$ & $2+\phi_h+\phi_v$ & Prop.~\ref{prop:margin_anatomy_asym} \\
\addlinespace
\multicolumn{3}{@{}l}{\emph{Profit and its drivers}}\\
$\pi(s)$ & Book profit per unit of handle & \eqref{eq:profit_bias_identity} \\
$Q(s)$ & Expected share of handle on the losing side & \eqref{eq:Q_def_methods} \\
$Q_{\mathrm{indep}}(s)$ & Value of $Q$ under Levitt independence & \eqref{eq:Q_indep} \\
$\delta(s)$ & $b_{h,L}-b_{h,W}$: outcome covariance of the bet share & \eqref{eq:delta_def} \\
$\theta(s)$ & Shading: deviation of the posted price from the truth & \eqref{eq:theta_lambda_def} \\
$\lambda(s)$ & Lean: deviation of the public's bet share from balance & \eqref{eq:theta_lambda_def} \\
$\mathcal{G}$ & The Goldilocks Zone & \eqref{eq:goldilocks_def} \\
\addlinespace
\multicolumn{3}{@{}l}{\emph{Public-belief model}}\\
$X$ & Latent expected outcome & \eqref{eq:belief_model_main} \\
$Y$ & Public's noisy, biased belief about $X$ & \eqref{eq:belief_model_main} \\
$\varepsilon$ & Systematic bias in the public's belief & \eqref{eq:belief_model_main} \\
$V,\,U$ & Belief noise and game randomness & \eqref{eq:belief_model_main} \\
\addlinespace
\multicolumn{3}{@{}l}{\emph{Estimation}}\\
$C_i$ & Realized outcome: $1$ if the home side won covered in game $i$ & \eqref{eq:pi_reconstructed} \\
$\pi_i$ & Book's profit margin for game $i$ & \eqref{eq:pi_reconstructed} \\
\bottomrule
\end{tabular}
\end{table}

\subsection*{S1. Asymmetric vigs, pushes, and the general profit--bias identity}

We allow for asymmetric vigs: if a unit stake is placed on the winning home side, the bettor receives a gross return of $1+\phi_h$, and if a unit stake is placed on the winning visitor side, the bettor receives a gross return of $1+\phi_v$, where $\phi_h,\phi_v>0$ represent the bettor's profit on a winning unit bet and satisfy $\phi_h\phi_v\le1$. The latter indicates that the book's overround $O\equiv(1+\phi_h)^{-1}+(1+\phi_v)^{-1}$ is at least one, and is exactly one only for a fair book. If the outcome is a push ($M=s$), all bets are returned. With total stake normalized so that $B_h(s)+B_v(s)=1$, the sportsbook profit for a single game is:
\begin{equation}
\label{eq:pi_def_asym}
\pi(M;s)=
\begin{cases}
B_v(s)-\phi_h\,B_h(s), & M>s,\\[2pt]
B_h(s)-\phi_v\,B_v(s), & M<s,\\[2pt]
0, & M=s.
\end{cases}
\end{equation}
Expected profit is computed by conditioning on the three outcomes $\{M<s\}$, $\{M>s\}$, $\{M=s\}$. If $M<s$ (home loses), then using \eqref{eq:pi_def_asym} and the definition of $b_{h,L}(s)$:
\[
\mathbb{E}[\pi(M;s)\mid M<s]
=
\mathbb{E}[B_h(s)-\phi_v(1-B_h(s))\mid M<s]
=
(1+\phi_v)\,b_{h,L}(s)-\phi_v.
\]
Similarly, if $M>s$, $\mathbb{E}[\pi(M;s)\mid M>s]=(1+\phi_h)\,b_{v,L}(s)-\phi_h$. If $M=s$, $\pi=0$. Taking expectations over the three cases yields:
\begin{equation}
\label{eq:profit_general}
\mathbb{E}[\pi(M;s)]
=
F_m(s)\,\big[(1+\phi_v)\,b_{h,L}(s)-\phi_v\big]
+
\bar F_m(s)\,\big[(1+\phi_h)\,b_{v,L}(s)-\phi_h\big].
\end{equation}
Define $Q_h(s)\equiv F_m(s)\,b_{h,L}(s)$, $Q_v(s)\equiv \bar F_m(s)\,b_{v,L}(s)$, $Q(s)\equiv Q_h(s)+Q_v(s)$, and $\Delta\phi\equiv \phi_v-\phi_h$. Substituting into \eqref{eq:profit_general} gives:
\[
\mathbb{E}[\pi(M;s)]
=
(1+\phi_h)\,Q(s)
+\Delta\phi\,Q_h(s)
-\phi_v\,F_m(s)-\phi_h\,\bar F_m(s).
\]
Using $F_m(s)+\pi_{\mathrm{push}}(s)+\bar F_m(s)=1$ and $\phi_v\,F_m(s)+\phi_h\,\bar F_m(s)=\phi_h(1-\pi_{\mathrm{push}}(s))+\Delta\phi\,F_m(s)$, we obtain the decomposition:
\begin{equation}
\label{eq:profit_decomp}
\mathbb{E}[\pi(M;s)]
=
(1+\phi_h)\,Q(s)
+\Delta\phi\,\big(Q_h(s)-F_m(s)\big)
-\phi_h\,\big(1-\pi_{\mathrm{push}}(s)\big).
\end{equation}
\emph{Special cases.} If $\phi_h=\phi_v=\phi$ then $\Delta\phi=0$ and \eqref{eq:profit_decomp} simplifies to:
\begin{equation}
\label{eq:profit_symmetric_push}
\mathbb{E}[\pi(M;s)]
=
(1+\phi)\,Q(s)
-\phi\,\big(1-\pi_{\mathrm{push}}(s)\big),
\end{equation}
which implies the uniform bound $\mathbb{E}[\pi(M;s)]\le(1+\phi)\,Q(s)$. Setting $\pi_{\mathrm{push}}(s)=0$ in \eqref{eq:profit_symmetric_push} recovers the profit--bias identity \eqref{eq:profit_bias_identity} of the main text.

\paragraph{Anatomy of the book's margin with asymmetric vigs.} The three-term decomposition of the main text (Prop.~\ref{prop:margin_anatomy}) can be generalized to asymmetric vigs by replacing the balanced-book reference $\tfrac12$ with the devigged price, and the symmetric hold $(1-\phi)/2$ by the overround-implied hold.

\begin{proposition}[Anatomy of the book's margin, asymmetric vigs]
\label{prop:margin_anatomy_asym}
Assume $\pi_{\mathrm{push}}(s)\equiv0$ and let $D\equiv2+\phi_h+\phi_v$. Define:
\[
p^\ast\equiv\frac{1+\phi_v}{D},
\qquad
h\equiv\frac{1-\phi_h\phi_v}{D},
\qquad
\theta(s)\equiv p^\ast-\bar F_m(s),
\qquad
\lambda(s)\equiv b_h(s)-p^\ast,
\]
where $p^\ast$ is the devigged implied win probability of the home side and $h$ is the book's hold. Then:
\begin{equation}
\label{eq:margin_anatomy_asym}
\mathbb{E}[\pi(s)]
= h
+ D\,\theta(s)\,\lambda(s)
+ D\,F_m(s)\bar F_m(s)\,\delta(s).
\end{equation}
Setting $\phi_h=\phi_v=\phi$ gives $p^\ast=\tfrac12$, $h=(1-\phi)/2$ and $D=2(1+\phi)$, recovering \eqref{eq:margin_anatomy_profit}.
\end{proposition}

\begin{proof}
Start from \eqref{eq:profit_general} with $\pi_{\mathrm{push}}(s)\equiv0$. By the law of total expectation, $b_h(s)=F_m(s)b_{h,L}(s)+\bar F_m(s)b_{h,W}(s)$; together with $\delta(s)=b_{h,L}(s)-b_{h,W}(s)$, this yields:
\[
b_{h,L}(s)=b_h(s)+\bar F_m(s)\,\delta(s),
\qquad
b_{h,W}(s)=b_h(s)-F_m(s)\,\delta(s),
\]
and since $B_v=1-B_h$, $b_{v,L}(s)=1-b_{h,W}(s)=\big(1-b_h(s)\big)+F_m(s)\,\delta(s)$. Substituting both into \eqref{eq:profit_general} and collecting the $\delta$ terms, whose coefficient is $F_m\bar F_m\big[(1+\phi_v)+(1+\phi_h)\big]=D\,F_m\bar F_m$, leads to:
\[
\mathbb{E}[\pi(s)]
= F_m\big[(1+\phi_v)b_h-\phi_v\big]
+ \bar F_m\big[(1+\phi_h)(1-b_h)-\phi_h\big]
+ D\,F_m\bar F_m\,\delta .
\]
Expanding the first two terms and factoring $b_h$ yields:
\[
\mathbb{E}[\pi(s)]
= \big[F_m(1+\phi_v)-\bar F_m(1+\phi_h)\big]\,b_h
\;+\;\big[\bar F_m-F_m\,\phi_v\big]
\;+\;D\,F_m\bar F_m\,\delta .
\]
We evaluate the two brackets in turn. Substituting $F_m=1-\bar F_m$ into the first and using $(1+\phi_v)+(1+\phi_h)=D$ and $(1+\phi_v)=Dp^\ast$:
\[
F_m(1+\phi_v)-\bar F_m(1+\phi_h)
=(1+\phi_v)-\bar F_m\big[(1+\phi_v)+(1+\phi_h)\big]
=D\,p^\ast-D\,\bar F_m
=D\,\theta .
\]
Performing the same substitution in the second bracket gives $\bar F_m-F_m\phi_v=\bar F_m(1+\phi_v)-\phi_v$. Writing $\bar F_m=p^\ast-\theta$, we have:
\[
\bar F_m(1+\phi_v)-\phi_v
= \big[p^\ast(1+\phi_v)-\phi_v\big]-\theta\,(1+\phi_v)
= \big[p^\ast(1+\phi_v)-\phi_v\big]-D\,\theta\,p^\ast .
\]
The bracketed constant evaluates to:
\[
p^\ast(1+\phi_v)-\phi_v
=\frac{(1+\phi_v)^2}{D}-\phi_v
=\frac{(1+\phi_v)^2-\phi_v D}{D}
=\frac{1-\phi_h\phi_v}{D},
\]
because $(1+\phi_v)^2-\phi_v(2+\phi_h+\phi_v)=1-\phi_h\phi_v$. Substituting both brackets back:
\[
\mathbb{E}[\pi(s)]
=D\,\theta(s)\,b_h(s)
+\frac{1-\phi_h\phi_v}{D}
-D\,\theta(s)\,p^\ast
+D\,F_m(s)\bar F_m(s)\,\delta(s),
\]
and collecting the two $\theta$ terms through $\lambda(s)=b_h(s)-p^\ast$:
\[
\mathbb{E}[\pi(s)]
=\frac{1-\phi_h\phi_v}{D}
+D\,\theta(s)\,\lambda(s)
+D\,F_m(s)\bar F_m(s)\,\delta(s).
\]

It remains to show that the constant $\frac{1-\phi_h\phi_v}{D}$ is equivalent to the book's hold $h$. Putting the two raw implied probabilities in the definition $O\equiv(1+\phi_h)^{-1}+(1+\phi_v)^{-1}$ over a common denominator results in the numerator $(1+\phi_v)+(1+\phi_h)=D$, such that:
\[
O=\frac{D}{(1+\phi_h)(1+\phi_v)},
\]
and hence:
\[
1-\frac{1}{O}
=\frac{D-(1+\phi_h)(1+\phi_v)}{D}
=\frac{1-\phi_h\phi_v}{D}
=h,
\]

where we have used $D-(1+\phi_h)(1+\phi_v)=(2+\phi_h+\phi_v)-(1+\phi_h+\phi_v+\phi_h\phi_v)=1-\phi_h\phi_v$. The constant term is therefore the share of handle that a balanced book keeps (i.e., $h$), giving \eqref{eq:margin_anatomy_asym}.
\end{proof}

\subsection*{S2. Recovering Levitt (2004) under independence}

The profit-bias identity \eqref{eq:profit_bias_identity} is a strict generalization of the bookmaker-profit expression in \citet{levitt2004gambling}. In Levitt's notation, $p$ is the probability that the favorite wins, $f$ is the fraction of total dollars wagered on the favorite, and $v$ is the vig; his expected gross profit per unit bet is:
\begin{equation}
\label{eq:levitt_profit}
\mathbb{E}[\pi_{\mathrm{Levitt}}]
=
\big[(1-p)\,f + p\,(1-f)\big](1+v) - \big[(1-p)(1-f) + p\,f\big].
\end{equation}
The first bracketed term is the expected fraction wagered on the eventual \emph{losing} side, from which the book collects $1+v$ per unit; the second bracketed term is the expected fraction on the eventual \emph{winning} side, on which it pays out $1$ per unit. Levitt's expression is therefore of the form $\mathrm{losing~share} \times(1+v) - \mathrm{winning~share}$, which follows as:
\[
\mathbb{E}[\pi_{\mathrm{Levitt}}]
=\big[(1-p)f+p(1-f)\big](1+v)-\Big(1-\big[(1-p)f+p(1-f)\big]\Big)
=(2+v)\,\big(f+p-2pf\big)-1 ,
\]
which is his Eq.~(2). His profit is thus already an affine, increasing function of a single losing-side share --- the same structural form as \eqref{eq:profit_bias_identity}.

An important distinction between Levitt's and our expressions pertains to the composition of the losing-side share. In \eqref{eq:levitt_profit}, $f$ is a single \emph{unconditional} number, such that forming $(1-p)f$ implicitly treats the favorite's bet share as statistically independent of the realized outcome. Identifying the favorite with the home side gives $p=\bar F_m(s)$ and $1-p=F_m(s)$, and imposing that independence gives:
\[
b_{h,L}(s)=b_h(s)=f,\qquad b_{v,L}(s)=b_v(s)=1-f ,
\]
such that our outcome-conditioned share \eqref{eq:Q_def_methods} equates to Levitt's bracket:
\[
Q(s)=F_m(s)\,f+\bar F_m(s)\,(1-f)=(1-p)f+p(1-f)=f+p-2pf\;\equiv\;Q_{\mathrm{indep}} .
\]
This shows alignment of bet shares, but not yet of profits. The two profit expressions are denominated in different units and use different vig conventions -- we reconcile the different conventions next, and then verify that the identity reproduces Levitt's Eq.~(2) exactly. We note also that Levitt's later step of writing $f=f(p)$ with $\partial f/\partial p>0$ does not affect his independence assumption: it describes how aggregate betting responds to the \emph{posted} line, not how bet shares covary with the \emph{realized} outcome.

\paragraph{Reconciling the two conventions.}
To allow a direct comparison between \eqref{eq:profit_bias_identity} and \eqref{eq:levitt_profit}, one must first reconcile two bookkeeping conventions relating to (i) the vig and (ii) the unit of measurement of profit.   


\emph{(i) The vig.} Levitt's $v$ is defined relative to the \emph{win} amount: a losing bettor pays
$1+v$ and a winning bettor collects $1$ (at the standard $-110$ price, $110$ is risked and $100$ is earned, so $v=1/10$). Our $\phi$ is the net payout per unit \emph{wagered}: a bettor risks $1$ to earn $\phi$. The relationship between Levitt's $v$ and our $\phi$ is thus given by:
\begin{equation}
\label{eq:vig_dictionary}
\phi=\frac{1}{1+v}.
\end{equation}
For example, $-110$ in American odds corresponds to $\phi=10/11$ and $v=1/10$. 

\emph{(ii) Normalization of profit.} Our $\pi(s)$ is the profit per unit of \emph{handle}: total stake is
normalized to $B_h(s)+B_v(s)=1$. Levitt's ``gross profit per unit bet'' is denominated in the
win amount, such that a single unit bet in his accounting puts $1+v$ of handle at risk. One Levitt unit
therefore carries $1+v=1/\phi$ units of handle, and the two profit measures are related by:
\begin{equation}
\label{eq:numeraire}
\mathbb{E}[\pi_{\mathrm{Levitt}}]=(1+v)\,\mathbb{E}[\pi(s)]=\frac{\mathbb{E}[\pi(s)]}{\phi}.
\end{equation}

\paragraph{Exact recovery of Levitt's Eq.~(2).}
Applying \eqref{eq:vig_dictionary} and \eqref{eq:numeraire} to the identity
\eqref{eq:profit_bias_identity} yields:
\begin{equation}
\label{eq:identity_levitt_units}
\frac{\mathbb{E}[\pi(s)]}{\phi}
=\frac{(1+\phi)\,Q(s)-\phi}{\phi}
=\Big(1+\tfrac{1}{\phi}\Big)Q(s)-1
=(2+v)\,Q(s)-1 .
\end{equation}
If we assume Levitt's independence and designate the favorite as the home side, \eqref{eq:identity_levitt_units} follows as:  
\[
\mathbb{E}[\pi_{\mathrm{Levitt}}]=(2+v)\,(f+p-2pf)-1,
\]
which is Levitt's Eq.~(2) exactly. 


As will be derived in the next section, the difference between Levitt's $f+p-2pf$ and our $Q(s)$ is precisely the covariance between bet shares and realized outcomes. Levitt's expression thus understates (or overstates) the book's margin whenever the public's money tends toward the eventual loser (winner).

\subsection*{S3. Derivation of the independence test}

Here we derive the statistical test of bet share--outcome independence \eqref{eq:delta_def}, which implies equality between conditional and unconditional bet shares:
\begin{equation}
    \label{eqn:eq_cond_uncond}
b_{h,L}(s)=b_h(s),
\qquad\qquad
b_{v,L}(s)=b_v(s).
\end{equation}
We show below that these two conditions collapse into a single equality, and thus bet share--outcome independence may be assayed with a single test. 


Consider the home condition first. Excluding pushes, the law of total expectation gives:
\begin{equation}
\label{eq:bh_total}
b_h(s)=F_m(s)\,b_{h,L}(s)+\bar F_m(s)\,b_{h,W}(s).
\end{equation}
The unconditional share is therefore a weighted average of two conditional shares, with the weights given by the probabilities of the two outcomes. Enforcing $b_{h,L}(s)=b_h(s)$ leads to:
\[
\bar F_m(s)\big[b_{h,L}(s)-b_{h,W}(s)\big]=0 .
\]
Assuming that $\bar F_m(s)>0$, this implies:
\[
b_{h,L}(s)=b_{h,W}(s).
\]

Note that the visitor loses the proposition precisely when the home side wins, such that:
\[
b_{v,L}(s)=1-b_{h,W}(s).
\]
Combined with $b_v(s)=1-b_h(s)$, the requirement $b_{v,L}(s)=b_v(s)$ is identical to $b_{h,W}(s)=b_h(s)$. In turn, $b_{h,W}(s)=b_h(s)$ implies:
\[
F_m(s)\big[b_{h,L}(s)-b_{h,W}(s)\big]=0 .
\]
Both conditions in \eqref{eqn:eq_cond_uncond} thus reduce to:
\[
\delta(s)=b_{h,L}(s)-b_{h,W}(s)=0.
\]
which is \eqref{eq:delta_def}.

\paragraph{Excess profit from bet-share outcome dependence}
Here we derive an expression for the change in the book profit that is brought about by a non-zero $\delta$. 
Substituting (\ref{eqn:eq_cond_uncond}) into \eqref{eq:Q_def_methods} yields the profit driver under Levitt's independence:
\begin{equation}
\label{eq:Q_indep}
Q_{\mathrm{indep}}(s)\equiv F_m(s)\,b_h(s)+\bar F_m(s)\,b_v(s).
\end{equation}
To relate $Q_{\mathrm{indep}}(s)$ to the realized losing-side share $Q(s)$, it is helpful to write both in terms of home bet shares:
\[
Q(s)=F_m(s)\,b_{h,L}(s)+\bar F_m(s)\big(1-b_{h,W}(s)\big),
\qquad
Q_{\mathrm{indep}}(s)=F_m(s)\,b_h(s)+\bar F_m(s)\big(1-b_h(s)\big).
\]
Subtracting $Q_{\mathrm{indep}}$ from $Q(s)$, one obtains:
\begin{equation}
    \label{eq:bh_Q_Qindep}
    Q(s)-Q_{\mathrm{indep}}(s)
    =F_m(s)\big[b_{h,L}(s)-b_h(s)\big]
    +\bar F_m(s)\big[b_h(s)-b_{h,W}(s)\big].
\end{equation}

Employing \eqref{eq:bh_total}, the two bracketed expressions above may be written as:
\begin{equation}
    \label{eq:bh_brackets}
    b_{h,L}(s)-b_h(s)=\bar F_m(s)\,\delta(s),
    \qquad\qquad
    b_h(s)-b_{h,W}(s)=F_m(s)\,\delta(s).  
\end{equation}
Substituting \eqref{eq:bh_brackets} back into \eqref{eq:bh_Q_Qindep} leads to:
\[
Q(s)-Q_{\mathrm{indep}}(s)
=2\,F_m(s)\bar F_m(s)\,\delta(s).
\]
To connect the excess profit to the covariance between bet shares and realized outcomes, let $I=\mathbf 1\{M<s\}$ represent the indicator that
the home side fails to cover, such that $\mathbb{E}[I]=F_m(s)$. Note that one can write $\mathbb{E}[B_h I]=F_m(s)\,b_{h,L}(s)$. The covariance between $B_h$ and $I$ is then given by:
\[
\mathrm{Cov}\big(B_h(s),I\big)
=\mathbb{E}[B_h I]-\mathbb{E}[B_h]\,\mathbb{E}[I]
=F_m(s)\big[b_{h,L}(s)-b_h(s)\big]
=F_m(s)\bar F_m(s)\,\delta(s),
\]
where in the last step we have used \eqref{eq:bh_brackets}. Putting all of this together, the excess profit can be written in two equivalent forms:
\begin{equation}
\label{eq:Q_excess}
Q(s)-Q_{\mathrm{indep}}(s)=2\,F_m(s)\,\bar F_m(s)\,\delta(s)
=2\,\mathrm{Cov}\!\big(B_h(s),\,\mathbf{1}\{M<s\}\big).
\end{equation}
The significance of \eqref{eq:Q_excess} is that the deviation of our profit driver $Q$ from Levitt's $Q_{\mathrm{indep}}$ is twice the covariance between bet shares and realized outcomes. 

\subsection*{S4. Public-belief model: \texorpdfstring{$Q$}{Q} as a misclassification probability}

Here we provide the proof of Proposition~\ref{prop:misclassification} in the no-push case where $\Pr(M=s)=0$. We model the public's belief about the latent expected margin $X$ with:
\begin{equation}
\label{eq:belief_model}
Y = X + \varepsilon + V,
\qquad
M = X + U,
\end{equation}
where $\varepsilon$ is a systematic bias and $U$ (game randomness) is independent of $V$ (belief noise). Let $F_U$ denote the CDF of $U$ and let $F_V$ denote the CDF of $V$.

Let the public bet the home side whenever $Y>s$. Define the Bernoulli indicator:
\[
B_h(s) \equiv \mathbf{1}\{Y>s\},
\qquad
B_v(s)\equiv 1-B_h(s)=\mathbf{1}\{Y\le s\}.
\]
This constrains the public's bet share to the two extrema $\{0,1\}$, such that the public acts on the signal $Y$ and stakes one side exclusively. This ``specialization'' converts the losing-side shares into misclassification probabilities but is not restrictive for $Q(s)$: with a continuum of bettors holding independent signals $Y_i=X+\varepsilon+V_i$, the home bet share $B_h(s)=\Pr(Y_i>s\mid X)=1-F_V(s-X-\varepsilon)\in(0,1)$ yields the same $Q(s)$ derived below, since integrating the fraction over $X$ reproduces \eqref{eq:Q_expectation} term by term.


\medskip
Recall the definition of $b_{h,L}(s)$ and $b_{v,L}(s)$ as conditional expected bet shares on the losing side:
\begin{equation}
\label{eq:mu_def}
b_{h,L}(s) \equiv \mathbb{E}[B_h(s)\mid M<s],
\qquad
b_{v,L}(s) \equiv \mathbb{E}[B_v(s)\mid M>s].
\end{equation}
Because $B_h(s)$ and $B_v(s)$ are Bernoulli, their conditional expectations may be written as conditional probabilities:
\begin{equation}
\label{eq:mu_as_probs}
b_{h,L}(s)=\Pr(Y>s\mid M<s),
\qquad
b_{v,L}(s)=\Pr(Y\le s\mid M>s).
\end{equation}

Substituting \eqref{eq:mu_as_probs} into the profit-bias identity \eqref{eq:Q_def_methods} yields:
\begin{equation}
\label{eq:Q_joint}
Q(s) = \Pr(Y>s,\; M<s) + \Pr(Y\le s,\; M>s).
\end{equation}
\begin{proof}[Proof of Proposition~\ref{prop:misclassification}]
The claim follows immediately from \eqref{eq:Q_joint}: the two events $\{Y>s,\,M<s\}$ (public backs home, home loses) and $\{Y\le s,\,M>s\}$ (public backs visitor, visitor loses) are exactly the two ways that the public's signal $Y$ misorders the realized margin of victory $M$ relative to the line $s$. The union of these disjoint events is the event that the public backs the losing side, such that $Q(s)=\Pr(\text{public backs the losing side})$, which is the probability of misclassifying the target $M>s$ using the predictor $\mathbf{1}\{Y>s\}$.
\end{proof}

\paragraph{Evaluating $Q(s)$ under the public-belief model.}
From \eqref{eq:belief_model}, the dependence between $Y$ and $M$ stems from the shared latent variable $X$. Even if $U\perp V$, the events $\{Y>s\}$ and $\{M<s\}$ are not generally independent; however, they are independent conditional on $X$ because $Y$ depends on $(X,V)$ and $M$ depends on $(X,U)$ with $U\perp V$.
We can therefore write:
\[
\Pr(Y>s,\;M<s)
=
\mathbb{E}_X\!\left[\Pr(Y>s\mid X)\Pr(M<s\mid X)\right],
\]
and similarly for $\Pr(Y\le s,\;M>s)$. Conditioning on $X$ gives:
\[
\Pr(Y>s\mid X)=\Pr(V>s-X-\varepsilon)=1-F_V(s-X-\varepsilon),
\]
and:
\[
\Pr(M<s\mid X)=\Pr(U<s-X)=F_U(s-X).
\]
Therefore, combining with \eqref{eq:Q_joint}:
\begin{align}
\label{eq:Q_expectation}
Q(s)
&=
\mathbb{E}_X\!\left[(1-F_V(s-X-\varepsilon))\,F_U(s-X)\right]
+
\mathbb{E}_X\!\left[F_V(s-X-\varepsilon)\,(1-F_U(s-X))\right].
\end{align}
Equation \eqref{eq:Q_expectation} is the expression for $Q(s)$ under the public-belief model.

To connect this to the losing-side bet shares, we write \eqref{eq:mu_as_probs} as:
\begin{equation}
    \label{eq:mu_as_probs_2}
    b_{h,L}(s)=\Pr(Y>s\mid M<s)=\frac{\Pr(Y>s, M<s)}{\Pr(M<s)}=\frac{\mathbb{E}_X\!\left[(1-F_V(s-X-\varepsilon))\,F_U(s-X)\right]}{F_m(s)},
\end{equation}
and
\begin{equation}
    \label{eq:mu_as_probs_3}
    b_{v,L}(s)=\Pr(Y\le s\mid M>s) = \frac{\Pr(Y\le s, M>s)}{\Pr(M>s)} = \frac{\mathbb{E}_X\!\left[F_V(s-X-\varepsilon)\,(1-F_U(s-X))\right]}{1-F_m(s)}.
\end{equation}
Substituting \eqref{eq:mu_as_probs_2} and \eqref{eq:mu_as_probs_3} into $Q(s)=b_{h,L}(s) F_m(s) + b_{v,L}(s)(1-F_m(s))$, we recover \eqref{eq:Q_expectation} identically.

\subsection*{S5. The Goldilocks book-profit region}

We define $\tau\equiv\phi/(1+\phi)$. The book is profitable if \eqref{eq:book_profit_condition} is $Q(s)>\tau$, and $Q(s)$ is affine in $F_m(s)$:
\begin{equation}
Q(s)=F_m(s)\,b_{h,L}(s)+\bar F_m(s)\,b_{v,L}(s)
=b_{v,L}(s)+F_m(s)\big(b_{h,L}(s)-b_{v,L}(s)\big),
\label{eq:Q_affine_in_F}
\end{equation}
such that \eqref{eq:book_profit_condition} is the linear inequality:
\begin{equation}
F_m(s)\big(b_{h,L}(s)-b_{v,L}(s)\big)
> \tau-b_{v,L}(s).
\label{eq:book_profit_condition_caseII}
\end{equation}
Below we prove Proposition~\ref{prop:goldilocks_brackets}.

\begin{proof}[Proof of Proposition~\ref{prop:goldilocks_brackets}]
Write $a=b_{h,L}(s)$, $b=b_{v,L}(s)$, $F=F_m(s)\in[0,1]$, so that $Q=(1-F)b+Fa$ is a convex combination of $a$ and $b$; consequently, $\min(a,b)\le Q\le\max(a,b)$, which is \eqref{eq:Q_sandwich}.

\emph{(i) Always.} If $b_{\min}(s)=\min(a,b)>\tau$, then $Q\ge\min(a,b)>\tau$ for every $F\in[0,1]$.

\emph{(ii) Never.} If $b_{\max}(s)=\max(a,b)\le\tau$, then $Q\le\max(a,b)\le\tau$ for every $F\in[0,1]$, so \eqref{eq:book_profit_condition} fails.

\emph{(iii) Threshold.} Suppose $\min(a,b)\le\tau<\max(a,b)$, such that $a\ne b$ and $F^\ast\equiv(\tau-b)/(a-b)$ is well defined. If $a>b$, then $b=\min(a,b)\le\tau<a=\max(a,b)$, so that $0\le\tau-b<a-b$; dividing by $a-b>0$ gives $F^\ast\in[0,1)$ and \eqref{eq:book_profit_condition_caseII} becomes the condition $F>F^\ast$. If instead $a<b$, then $a\le\tau<b$, so that $a-b\le\tau-b<0$; dividing by $a-b<0$ gives $F^\ast\in(0,1]$, and the direction of the inequality \eqref{eq:book_profit_condition_caseII} is reversed to yield $F<F^\ast$. In both cases, $F^\ast\in[0,1]$ and \eqref{eq:book_profit_condition_caseII} places $F_m(s)$ on the side of $F^\ast$ that assigns greater weight to $b_{\max}(s)$. 

The sufficiency of $b_{\min}(s)>\tau$ is case (i); the necessity of $b_{\max}(s)>\tau$ follows from case (ii).
\end{proof}

\noindent When the shares coincide ($a=b=b_L$), \eqref{eq:Q_affine_in_F} reduces to $Q(s)=b_L(s)$ and the book profit is positive if $b_L(s)>\tau$. Intersecting the book-profit region with the bettor set $\mathcal{B}$ gives the Goldilocks Zone $\mathcal{G}=\{s:Q(s)>\tau\}\cap\mathcal{B}$, with $Q(s)$ given by \eqref{eq:Q_affine_in_F}. If $F_m$ is strictly increasing, bettor profitability \eqref{eq:bettor_set} may be written in terms of quantiles:
\begin{equation}
\mathcal{B}
= \Big(-\infty,\;F_m^{-1}(\tau)\Big) \;\cup\; \Big(F_m^{-1}(1-\tau),\;\infty\Big).
\label{eq:bettor_set_quantiles}
\end{equation}
The sufficient and necessary conditions of the main text follow by intersecting cases (i) and (ii) with membership of $F_m(s)$ in one of these tails; only the Threshold regime (iii) requires the joint values of $F_m$, $b_{h,L}$, and $b_{v,L}$ at the same $s$, and is therefore evaluated empirically.

%% file: template.bbl
\begin{thebibliography}{41}
\providecommand{\natexlab}[1]{#1}
\providecommand{\url}[1]{\texttt{#1}}
\expandafter\ifx\csname urlstyle\endcsname\relax
  \providecommand{\doi}[1]{doi: #1}\else
  \providecommand{\doi}{doi: \begingroup \urlstyle{rm}\Url}\fi

\bibitem[Avery and Chevalier(1999)]{avery1999sentiment}
Christopher Avery and Judith Chevalier.
\newblock Identifying investor sentiment from price paths: The case of football
  betting.
\newblock \emph{The Journal of Business}, 72\penalty0 (4):\penalty0 493--521,
  1999.
\newblock \doi{10.1086/209625}.

\bibitem[Croxson and Reade(2014)]{croxson2014information}
Karen Croxson and J.~James Reade.
\newblock Information and efficiency: Goal arrival in soccer betting.
\newblock \emph{The Economic Journal}, 124\penalty0 (575):\penalty0 62--91,
  2014.
\newblock \doi{10.1111/ecoj.12033}.

\bibitem[Dmochowski(2023)]{dmochowski2023statistical}
Jacek~P Dmochowski.
\newblock A statistical theory of optimal decision-making in sports betting.
\newblock \emph{Plos one}, 18\penalty0 (6):\penalty0 e0287601, 2023.

\bibitem[{DraftKings Network}(2026)]{dknetwork_splits}
{DraftKings Network}.
\newblock Betting splits.
\newblock
  \url{https://dknetwork.draftkings.com/draftkings-sportsbook-betting-splits/},
  2026.
\newblock Public handle and ticket percentages by market; accessed 2026.

\bibitem[{ESPN}(2026)]{espn_scoreboard}
{ESPN}.
\newblock Scoreboard api.
\newblock \url{https://site.api.espn.com/apis/site/v2/sports/}, 2026.
\newblock Final scores and game status; accessed 2026.

\bibitem[Franck et~al.(2011)Franck, Verbeek, and
  N\"{u}esch]{franck2011sentimental}
Egon~P. Franck, Erwin Verbeek, and Stephan N\"{u}esch.
\newblock Sentimental preferences and the organizational regime of betting
  markets.
\newblock \emph{Southern Economic Journal}, 78\penalty0 (2):\penalty0 502--518,
  2011.
\newblock \doi{10.4284/0038-4038-78.2.502}.

\bibitem[Gandar et~al.(1998)Gandar, Dare, Brown, and Zuber]{gandar1998informed}
John~M. Gandar, William~H. Dare, Craig~R. Brown, and Richard~A. Zuber.
\newblock Informed traders and price variations in the betting market for
  professional basketball games.
\newblock \emph{The Journal of Finance}, 53\penalty0 (1):\penalty0 385--401,
  1998.
\newblock \doi{10.1111/0022-1082.155346}.

\bibitem[Gandar et~al.(2002)Gandar, Zuber, Johnson, and
  Dare]{gandar2002reexamining}
John~M. Gandar, Richard~A. Zuber, R.~Stafford Johnson, and William Dare.
\newblock Re-examining the betting market on major league baseball games: Is
  there a reverse favourite-longshot bias?
\newblock \emph{Applied Economics}, 34\penalty0 (10):\penalty0 1309--1317,
  2002.
\newblock \doi{10.1080/00036840110095427}.

\bibitem[Gray and Gray(1997)]{gray1997testing}
Philip~K. Gray and Stephen~F. Gray.
\newblock Testing market efficiency: Evidence from the {NFL} sports betting
  market.
\newblock \emph{The Journal of Finance}, 52\penalty0 (4):\penalty0 1725--1737,
  1997.
\newblock \doi{10.1111/j.1540-6261.1997.tb01129.x}.

\bibitem[Hub\'{a}\v{c}ek et~al.(2019)Hub\'{a}\v{c}ek, \v{S}ourek, and
  \v{Z}elezn\'{y}]{hubacek2019exploiting}
Ond\v{r}ej Hub\'{a}\v{c}ek, Gustav \v{S}ourek, and Filip \v{Z}elezn\'{y}.
\newblock Exploiting sports-betting market using machine learning.
\newblock \emph{International Journal of Forecasting}, 35\penalty0
  (2):\penalty0 783--796, 2019.
\newblock \doi{10.1016/j.ijforecast.2019.01.001}.

\bibitem[Humphreys(2010)]{humphreys2010pointspread}
Brad~R. Humphreys.
\newblock Point spread shading and behavioral biases in {NBA} betting markets.
\newblock \emph{Rivista di Diritto ed Economia dello Sport}, 6\penalty0
  (1):\penalty0 13--26, 2010.

\bibitem[Humphreys(2011)]{humphreys2011financial}
Brad~R. Humphreys.
\newblock The financial consequences of unbalanced betting on {NFL} games.
\newblock \emph{International Journal of Sport Finance}, 6\penalty0
  (1):\penalty0 60--71, 2011.

\bibitem[Humphreys et~al.(2013)Humphreys, Paul, and
  Weinbach]{humphreys2013homeunderdog}
Brad~R. Humphreys, Rodney~J. Paul, and Andrew~P. Weinbach.
\newblock Bettor biases and the ``home-underdog'' bias in the {NFL}.
\newblock \emph{International Journal of Sport Finance}, 8\penalty0
  (4):\penalty0 294--311, 2013.

\bibitem[Jord{\`a}(2005)]{jorda2005estimation}
{\`O}scar Jord{\`a}.
\newblock Estimation and inference of impulse responses by local projections.
\newblock \emph{American Economic Review}, 95\penalty0 (1):\penalty0 161--182,
  2005.

\bibitem[Jullien and Salani{\'e}(2000)]{jullien2000estimating}
Bruno Jullien and Bernard Salani{\'e}.
\newblock Estimating preferences under risk: The case of racetrack bettors.
\newblock \emph{Journal of Political Economy}, 108\penalty0 (3):\penalty0
  503--530, 2000.
\newblock \doi{10.1086/262127}.

\bibitem[Krieger and Fodor(2013)]{krieger2013price}
Kevin Krieger and Andy Fodor.
\newblock Price movements and the prevalence of informed traders: The case of
  line movement in college basketball.
\newblock \emph{Journal of Economics and Business}, 68:\penalty0 70--82, 2013.
\newblock \doi{10.1016/j.jeconbus.2013.04.001}.

\bibitem[Kuypers(2000)]{kuypers2000information}
Tim Kuypers.
\newblock Information and efficiency: An empirical study of a fixed odds
  betting market.
\newblock \emph{Applied Economics}, 32\penalty0 (11):\penalty0 1353--1363,
  2000.
\newblock \doi{10.1080/00036840050151449}.

\bibitem[Levitt(2004)]{levitt2004gambling}
Steven~D. Levitt.
\newblock Why are gambling markets organised so differently from financial
  markets?
\newblock \emph{The Economic Journal}, 114\penalty0 (495):\penalty0 223--246,
  2004.
\newblock \doi{10.1111/j.1468-0297.2004.00207.x}.
\newblock URL
  \url{https://pricetheory.uchicago.edu/levitt/Papers/LevittWhyAreGamblingMarkets2004.pdf}.

\bibitem[Moskowitz(2021)]{moskowitz2021asset}
Tobias~J. Moskowitz.
\newblock Asset pricing and sports betting.
\newblock \emph{The Journal of Finance}, 76\penalty0 (6):\penalty0 3153--3209,
  2021.
\newblock \doi{10.1111/jofi.13082}.

\bibitem[Ottaviani and S{\o}rensen(2008)]{ottaviani2008favorite}
Marco Ottaviani and Peter~Norman S{\o}rensen.
\newblock The favorite-longshot bias: An overview of the main explanations.
\newblock In Donald~B. Hausch and William~T. Ziemba, editors, \emph{Handbook of
  Sports and Lottery Markets}, pages 83--101. North-Holland, 2008.

\bibitem[Paul and Weinbach(2002)]{paul2002totals}
Rodney~J. Paul and Andrew~P. Weinbach.
\newblock Market efficiency and a profitable betting rule: Evidence from totals
  on professional football.
\newblock \emph{Journal of Sports Economics}, 3\penalty0 (3):\penalty0
  256--263, 2002.
\newblock \doi{10.1177/1527002502003003003}.

\bibitem[Paul and Weinbach(2007)]{paul2007sportsbook}
Rodney~J. Paul and Andrew~P. Weinbach.
\newblock Does sportsbook.com set pointspreads to maximize profits? tests of
  the levitt model of sportsbook behavior.
\newblock \emph{Journal of Prediction Markets}, 1\penalty0 (3):\penalty0
  209--218, 2007.
\newblock URL
  \url{https://www.ubplj.org/index.php/jpm/article/download/429/461}.

\bibitem[Paul and Weinbach(2008)]{paul2008nba}
Rodney~J. Paul and Andrew~P. Weinbach.
\newblock Price setting in the nba gambling market: Tests of the levitt model
  of sportsbook behavior.
\newblock \emph{International Journal of Sport Finance}, 3\penalty0
  (3):\penalty0 137--145, 2008.
\newblock URL
  \url{https://fitpublishing.com/content/price-setting-nba-gambling-market-tests-levitt-model-sportsbook-behavior-pp-137-145}.

\bibitem[Paul and Weinbach(2009)]{paul2009ncaa}
Rodney~J. Paul and Andrew~P. Weinbach.
\newblock Sportsbook behavior in the {NCAA} football betting market: Tests of
  the traditional and {Levitt} models of sportsbook behavior.
\newblock \emph{The Journal of Prediction Markets}, 3\penalty0 (2):\penalty0
  21--37, 2009.

\bibitem[Paul and Weinbach(2011)]{paul2011nfl}
Rodney~J. Paul and Andrew~P. Weinbach.
\newblock {NFL} bettor biases and price setting: Further tests of the {Levitt}
  hypothesis of sportsbook behaviour.
\newblock \emph{Applied Economics Letters}, 18\penalty0 (2):\penalty0 193--197,
  2011.
\newblock \doi{10.1080/13504850903508242}.

\bibitem[Paul and Weinbach(2012)]{paul2012nhl}
Rodney~J. Paul and Andrew~P. Weinbach.
\newblock Sportsbook pricing and the behavioral biases of bettors in the {NHL}.
\newblock \emph{Journal of Economics and Finance}, 36\penalty0 (1):\penalty0
  123--135, 2012.
\newblock \doi{10.1007/s12197-009-9112-4}.

\bibitem[Sauer(1998)]{sauer1998economics}
Raymond~D. Sauer.
\newblock The economics of wagering markets.
\newblock \emph{Journal of Economic Literature}, 36\penalty0 (4):\penalty0
  2021--2064, 1998.

\bibitem[Shin(1991)]{shin1991optimal}
Hyun~Song Shin.
\newblock Optimal betting odds against insider traders.
\newblock \emph{The Economic Journal}, 101\penalty0 (408):\penalty0 1179--1185,
  1991.
\newblock \doi{10.2307/2234434}.

\bibitem[Shin(1992)]{shin1992prices}
Hyun~Song Shin.
\newblock Prices of state contingent claims with insider traders, and the
  favourite-longshot bias.
\newblock \emph{The Economic Journal}, 102\penalty0 (411):\penalty0 426--435,
  1992.
\newblock \doi{10.2307/2234526}.

\bibitem[Shin(1993)]{shin1993measuring}
Hyun~Song Shin.
\newblock Measuring the incidence of insider trading in a market for
  state-contingent claims.
\newblock \emph{The Economic Journal}, 103\penalty0 (420):\penalty0 1141--1153,
  1993.
\newblock \doi{10.2307/2234240}.

\bibitem[Shleifer and Vishny(1997)]{shleifer1997limits}
Andrei Shleifer and Robert~W. Vishny.
\newblock The limits of arbitrage.
\newblock \emph{The Journal of Finance}, 52\penalty0 (1):\penalty0 35--55,
  1997.
\newblock \doi{10.1111/j.1540-6261.1997.tb03807.x}.

\bibitem[Simon(2024)]{simon2024inefficient}
Jay Simon.
\newblock Inefficient forecasts at the sportsbook: An analysis of real-time
  betting line movement.
\newblock \emph{Management Science}, 70\penalty0 (12):\penalty0 8583--8611,
  2024.
\newblock \doi{10.1287/mnsc.2022.00456}.

\bibitem[Snowberg and Wolfers(2010)]{snowberg2010explaining}
Erik Snowberg and Justin Wolfers.
\newblock Explaining the favorite-long shot bias: Is it risk-love or
  misperceptions?
\newblock \emph{Journal of Political Economy}, 118\penalty0 (4):\penalty0
  723--746, 2010.
\newblock \doi{10.1086/655844}.

\bibitem[Snowberg et~al.(2013)Snowberg, Wolfers, and
  Zitzewitz]{snowberg2013prediction}
Erik Snowberg, Justin Wolfers, and Eric Zitzewitz.
\newblock Prediction markets for economic forecasting.
\newblock In Graham Elliott and Allan Timmermann, editors, \emph{Handbook of
  Economic Forecasting}, volume~2, pages 657--687. Elsevier, 2013.

\bibitem[Stan\v{e}k(2017)]{stanek2017homebias}
Rostislav Stan\v{e}k.
\newblock Home bias in sport betting: Evidence from {C}zech betting market.
\newblock \emph{Judgment and Decision Making}, 12\penalty0 (2):\penalty0
  168--172, 2017.

\bibitem[Thaler and Ziemba(1988)]{thaler1988anomalies}
Richard~H. Thaler and William~T. Ziemba.
\newblock Anomalies: Parimutuel betting markets: Racetracks and lotteries.
\newblock \emph{Journal of Economic Perspectives}, 2\penalty0 (2):\penalty0
  161--174, 1988.
\newblock \doi{10.1257/jep.2.2.161}.

\bibitem[Vandenbruaene et~al.(2022)Vandenbruaene, De~Ceuster, and
  Annaert]{vandenbruaene2022efficient}
Jonas Vandenbruaene, Marc De~Ceuster, and Jan Annaert.
\newblock Efficient spread betting markets: A literature review.
\newblock \emph{Journal of Sports Economics}, 23\penalty0 (7):\penalty0
  907--949, 2022.
\newblock \doi{10.1177/15270025211071042}.

\bibitem[Whelan(2024)]{whelan2024risk}
Karl Whelan.
\newblock Risk aversion and favourite--longshot bias in a competitive
  fixed-odds betting market.
\newblock \emph{Economica}, 91\penalty0 (361):\penalty0 188--209, 2024.
\newblock \doi{10.1111/ecca.12500}.

\bibitem[Whelan(2025)]{whelan2025estimates}
Karl Whelan.
\newblock On estimates of insider trading in sports betting.
\newblock \emph{The Manchester School}, 2025.
\newblock \doi{10.1111/manc.12505}.

\bibitem[Wolfers and Zitzewitz(2004)]{wolfers2004prediction}
Justin Wolfers and Eric Zitzewitz.
\newblock Prediction markets.
\newblock \emph{Journal of Economic Perspectives}, 18\penalty0 (2):\penalty0
  107--126, 2004.
\newblock \doi{10.1257/0895330041371321}.

\bibitem[Woodland and Woodland(1994)]{woodland1994baseball}
Linda~M. Woodland and Bill~M. Woodland.
\newblock Market efficiency and the favorite-longshot bias: The baseball
  betting market.
\newblock \emph{The Journal of Finance}, 49\penalty0 (1):\penalty0 269--279,
  1994.
\newblock \doi{10.1111/j.1540-6261.1994.tb04429.x}.

\end{thebibliography}
